\documentclass[12pt,english]{article}
\usepackage{amsmath}
\usepackage[linesnumbered,ruled]{algorithm2e}
\usepackage{amsfonts}
\usepackage{mathrsfs}
\usepackage{amsthm}
\usepackage{amssymb}
\usepackage{blkarray}
\usepackage{bbm}
\usepackage{commath}
\usepackage[toc, page]{appendix}
\usepackage{mathtools}
\usepackage{thmtools}
\usepackage{thm-restate}
\usepackage{changepage}
\usepackage[cal=boondox]{mathalfa}
\usepackage{autobreak}
\usepackage{overpic}
\usepackage[margin=1cm]{caption}

\makeatletter
\newcommand{\algorithmfootnote}[2][\footnotesize]{%
	\let\old@algocf@finish\@algocf@finish
	\def\@algocf@finish{\old@algocf@finish
		\leavevmode\rlap{\begin{minipage}{\linewidth}
				#1#2
			\end{minipage}}%
		}%
	}
	\makeatother
	
	\newcommand{\topbottombraced}[2]{%
  \raise.5ex\vtop{
    \vbox{%
      \hbox{$\left\{\begin{tabular}{@{}l@{}}#1\end{tabular}\right\}$}
      \vskip1pt
    }
    \vbox{%
      \vskip1pt
      \hbox{$\left\{\begin{tabular}{@{}l@{}}#2\end{tabular}\right\}$}
    }
  }%
}
	
\let\oldnl\nl
\newcommand{\nonl}{\renewcommand{\nl}{\let\nl\oldnl}}

\usepackage{multirow}
\usepackage{easy-todo}
\usepackage{accents}
\usepackage[font=small,labelfont=bf]{caption}
\usepackage{lscape}
\usepackage[normalem]{ulem}
\useunder{\uline}{\ul}{}

\usepackage[T1]{fontenc}
\usepackage[utf8]{inputenc}
\usepackage{pdflscape}
\usepackage{enumitem}
\usepackage{float}
\usepackage[hang, flushmargin]{footmisc}

\usepackage{natbib}
\usepackage{booktabs}
\usepackage{tikz}
\usetikzlibrary{patterns}
\usepackage{graphicx}
\usepackage{color}
\usepackage{hyperref}
\hypersetup{colorlinks=true,linkcolor=blue, citecolor=blue, filecolor=blue,urlcolor=blue}
\usepackage[graphicx]{realboxes}

\usepackage[all]{hypcap}
\usepackage[english]{babel}
\usepackage{threeparttable}
\usepackage{eurosym}

\usepackage[doublespacing]{setspace}
\usepackage{etoolbox}
\usepackage{cleveref}
\makeatletter
\patchcmd\TPTdoTablenotes{%
	\TPTnoteSettings
}{%
	\TPTnoteSettings
	\setlength{\itemsep}{2ex}%
}{}{\errmessage{Patching \noexpand\TPTdoTablenotes failed}}
\makeatother

\usepackage{subfiles}

\makeatletter
\patchcmd\TPTdoTablenotes{%
	\TPTnoteSettings
}{%
\TPTnoteSettings
\setlength{\itemsep}{2ex}%
}{}{\errmessage{Patching \noexpand\TPTdoTablenotes failed}}
\makeatother

\makeatletter

\patchcmd{\NAT@test}{\else \NAT@nm}{\else \NAT@nmfmt{\NAT@nm}}{}{}

\DeclareRobustCommand\citepos
{\begingroup
	\let\NAT@nmfmt\NAT@posfmt
	\NAT@swafalse\let\NAT@ctype\z@\NAT@partrue
	\@ifstar{\NAT@fulltrue\NAT@citetp}{\NAT@fullfalse\NAT@citetp}}

\let\NAT@orig@nmfmt\NAT@nmfmt
\def\NAT@posfmt#1{\NAT@orig@nmfmt{#1's}}

\makeatother

\def\sym#1{\ifmmode^{#1}\else\(^{#1}\)\fi}

\newcommand{\pdv}[2]{\frac{\partial #1}{\partial #2}}

\newcommand{\ubar}[1]{\underaccent{\bar}{#1}}

\declaretheorem[name=Proposition]{prop}

\declaretheorem[name=Definition]{defn}

\usepackage{titling}

\usepackage{geometry}
\renewcommand\footnotelayout{
  \advance\leftskip 0.4in
  \advance\rightskip 0.4in
 } 
 
\begin{document}
\setlength{\droptitle}{-7em}   
\title{The Economics of Orbit Use: \\ Open Access, External Costs, and Runaway Debris Growth\thanks{We are grateful to Dan Kaffine, Jon Hughes, Martin Boileau, Miles Kimball, Alessandro Peri, Matt Burgess, Sami Dakhlia, S{\'e}bastien Rouillon, Martin Abel, David Munro, Derek Lemoine, and many seminar participants for helpful comments and feedback. We are especially grateful to Aditya Jain for excellent research assistance. Funding for this research was generously provided by Center for Advancement of Teaching and Research in Social Science and the Reuben A Zeubrow Fellowship in Economics at CU Boulder. All errors are our own.}}

\author{
        Akhil Rao\\ \small Middlebury College\thanks{Department of Economics, Warner Hall, 303 College Street, Middlebury College, 05753; akhilr@middlebury.edu}
        \and 
        Giacomo Rondina 
        \\ \small UC San Diego\thanks{Department of Economics, 9500 Gilman Dr, La Jolla, CA 92093; grondina@ucsd.edu}
        }

\maketitle
\vspace{-0.75cm}
\abstract{ 
We present a dynamic physico-economic model of Earth orbit use with endogenous satellite collision risk to study conditions under which debris-producing collisions between orbiting bodies result in debris growth that may render Earth's orbits unusable, an outcome known as Kessler Syndrome.  We characterize the dynamics of objects in orbit under open access as well as when external costs---the impact of an additional satellite launch on the collision risk faced by all satellites---are internalized, and we show that Kessler Syndrome can emerge in both cases. Finally, we show that once the economic incentives of satellite launching are modeled, for Kessler Syndrome to emerge, autocatalytic debris growth is essential. In our main calibration, Kessler Syndrome can emerge anytime between the year 2040 and the year 2184, with the precise date being very sensitive to the calibration of autocatalytic debris growth parameters.} 

JEL codes: Q20, Q54, Q57, D62 \par
Keywords: open-access commons, satellites, space debris, dynamic externality, tipping point

\newpage
\section{Introduction}

Satellite services are increasingly important in the modern world. As humans launch more satellites, the risk of collisions between orbiting objects increases. Such collisions can destroy satellites and produce orbital debris, further increasing the risk of future collisions. The worst-case scenario is runaway debris growth, known as Kessler Syndrome, wherein the production of debris due to collisions between orbiting bodies becomes self-sustaining and irreversible. Kessler Syndrome will render valuable regions of orbital space unusable for decades, centuries, or longer. To make matters worse, international treaty law places orbital space under an open-access regime: anyone can place a satellite in any orbit they choose. 

How will open access affect orbital debris accumulation, collision risk, and the occurrence of Kessler Syndrome? What are the external costs that create a coordination problem under open access? Can Kessler Syndrome result even if external costs are fully internalized?  To address these questions, we study a dynamic model that combines the economic intertemporal decision problem of satellite operators with the physics of orbiting objects. 

We isolate a key economic condition that determines the optimal satellite launch rate by equating the cost of launching to its return discounted by an ``effective discount rate,'' which reflects the time-varying collision risk. As satellites and debris jointly evolve over time, so does the collision risk and the resulting launch rate, establishing a feedback effect where the economic incentives shape the physical dynamics. Under open access, satellite operators do not internalize the full impact of their launching decisions on the collision risk, leading to an inefficiently high launch rate. We show that the resulting external cost can be represented as an upward adjustment to the effective discount rate that a social planner would implement in choosing the size of an optimal satellite fleet.

Our analysis uncovers three main findings. First, under open access, we show that Kessler Syndrome is an important concern only when debris growth is self-reinforcing (also known as ``autocatalytic growth''), that is when debris-to-debris collisions are a substantial contributor to the growth of objects in orbit. Absent such contribution, the intertemporal considerations in the optimal launching decisions---the core feature of the economic side of our model---typically lead to a non-degenerative path of satellites and debris. Second, we show that Kessler Syndrome can result even when the external costs, represented by the impact of private launching decisions on collision risk and debris growth, are fully internalized. Intuitively, the short-run returns of a larger satellite fleet might dominate the long-run returns of a smaller fleet, which is once again the consequence of the intertemporal tradeoff introduced by the economic model. Third, we show how the timing of Kessler Syndrome under open access varies with the rate of autocatalytic debris growth and the growth rate of satellite returns. According to our calculations, under empirically plausible parameterizations, Kessler Syndrome can occur as early as the year 2040 and as late as the year 2184. The large range underscores the sensitivity of the system to autocatalytic debris growth dynamics once the economic incentives to launch are taken into consideration, which is the central result of our theoretical analysis.

\noindent \textsc{Institutional Background.} \; To motivate the institutional features of our model, we describe some institutional and physical features of orbit use. The Outer Space Treaty of 1967 (OST) is the primary international legal framework governing orbit use. It prevents establishment of orbital property rights, rendering orbital space a fixed facility subject to unrestricted access.\footnote{Specifically, Article 2 of the OST forbids national appropriation or claims of sovereignty over outer space, prohibiting national authorities from establishing orbital property rights \citep{gorove1969}.} Lacking the ability to legally exclude each other from orbital ``slots,'' satellites tend to cluster near valuable slots and risk colliding with each other. Open access thus induces an externality similar to congestion in a high-seas fishery \citep{gordon1954}.\footnote{\citet{ostrom1999} provides some insight into why decentralized orbit management may face challenges. Orbit users are a diverse and international group, ranging from national militaries and intelligence agencies to corporations, universities, and wealthy individuals. Physically excluding potential users from orbital regions without creating additional debris is difficult, and the relevant conflict resolution mechanisms are unclear. National regulatory regimes must also contend with ``launch leakage,'' which has happened at least once already \citep{wapoSats}. \citet{weedenchow2012} discuss some of these issues in more detail.} 

Satellites also produce debris over their lifecycle. Launches leave spent rocket stages and separation bolts in orbit, satellite operations leave paint chips and tools behind, and satellites which are not deorbited or shifted to disposal orbits at the end of their life become debris themselves.
Satellites struck by debris can shatter into hundreds of hazardous debris fragments.\footnote{Objects in orbit move at velocities higher than 5 km/s, so even debris as small as 10 cm in diameter can be hazardous to active satellites. Debris as low as 900 km above the Earth's surface can take centuries to deorbit naturally, while debris at 36,000 km can take even longer \citep{weedenIAC}. Currently, there are more than 2,000 operating satellites in orbit, up to 600,000 pieces of debris large enough to cause satellite loss, and millions of smaller particles that can degrade satellite performance \citep{ailoretal}. } Compounding the problem, collisions between debris objects can generate even more hazardous debris.
Debris accumulation can cause a cascading series of collisions between orbital objects, creating a growing debris field which can render orbital regions unusable and impassable for decades or centuries. This phenomenon is known as ``Kessler Syndrome'' \citep{kessler1}.\footnote{Kessler Syndrome can cause large economic losses, both directly from damage to active satellites and indirectly from limiting access to space \citep{bradleywein2009, schaubetal2015}. Existing estimates of debris growth indicate the risk of Kessler Syndrome is highest in low-Earth orbit (LEO), where it threatens current and future imaging and telecommunications satellites and can reduce access to higher orbits \citep{kessler2, lewis2020understanding}. In the worst-case scenario Kessler Syndrome could completely block human access to space, marking an eventual end to services like GPS and satellite imaging. Long-run disruption of satellite services will make it harder to measure economic activity, reduce weather-related uncertainty, measure environmental degradation and respond to natural disasters, monitor environmental policy compliance, and meet conservation goals \citep{nistGPS, donaldson2016view, sullivan2018using, baragwanath2019detecting, jain2020benefits, cooke2020market, stroming2020quantifying, bernknopf2021earth}.}

\noindent \textsc{Connection to the Literature.} \; The interaction between the physical dynamics and the intertemporal economic incentives to launching are underexplored in the literature on orbit use.\footnote{\cite{wienzierl2018} highlights several issues in the development of a space economy, from space debris to holdup problems and market design, which economists are uniquely positioned to address. Our focus is solely on the \emph{environmental} problems in orbit associated with existing institutions.} We contribute to the literature on decentralized and optimal responses to interacting dynamic externalities by showing how profit-maximizing responses to congestion can limit pollution production, and physico-economic conditions under which profit maximization leads to runaway pollutant accumulation.\footnote{While open access problems have been well-studied in terrestrial settings such as fisheries, forests, climate, oil fields, traffic, and invasive species management, open access to orbital resources is not as well understood. Though results from these other settings provide some helpful intuition, open access and the orbital mechanics governing collision risk and debris production create unique physico-economic feedback loops.} Though prior economic literature on orbit use has identified and quantified the externalities driving collisional rent dissipation \citep{aac2015, rouillon2019, raoetal2020, bealetal2020}, the channels through which the externality operates have not been explored in detail. Further, the economics of Kessler Syndrome remain understudied, with prior literature either explicitly neglecting it (e.g. \citet{rouillon2019, raoetal2020}) or analyzing it under conditions which rule out its occurrence in equilibrium (e.g. \citet{aac2018}).\footnote{Specifically, \citet{aac2018} consider a model environment in which open access will never cause Kessler Syndrome, since orbits become economically unprofitable before they become physically unusable (``economic Kessler Syndrome''). The differences between our results are due to definitions of Kessler Syndrome and physical generality in the model setup. \citet{aac2018} consider a definition of Kessler Syndrome where satellites are destroyed with probability one (``unusable orbits'') and disallow collisions between debris objects. We define Kessler Syndrome as states from which the debris stock diverges to infinity (``runaway debris growth'') and allow collisions between debris objects. The latter feature is critical for understanding orbit-use dynamics, as collisions between debris are becoming increasingly likely and are expected to dominate the long-run dynamics of the orbital environment \citep{wapoCollision2020, lewis2020understanding}. Our definition encompasses the one in \cite{aac2018} while allowing Kessler Syndrome to occur over time due to dynamic feedbacks---analogous to crossing a critical dispensation threshold in a fishery with Allee effects.} Our framework unifies existing research on orbit use, showing model features necessary to obtain different outcomes and enabling calibration and quantitative analysis. We also identify an important difference between resource-use dynamics when the relevant capital stock is provided by nature (e.g. fish in fisheries) and when it is provided by humans (e.g. satellites in orbits): while increases in the discount rate make equilibrium and optimal resource collapse more likely for natural capital, they have the opposite effect for artificial capital.

The rest of the paper is organized as follows. In section \ref{sec:simple_econ} we present a simple model that develops the essential economic intuition for orbit use. In section \ref{sec:physics_only} we present a physical model which captures the key elements of orbits as a resource and shows how Kessler Syndrome can occur. Section \ref{sec:full-econ-model} develops the dynamic economic model of open-access and optimal orbit use in the general physical setting, decomposing the channels through which the externality operates, and illustrating how Kessler Syndrome shapes the dynamics of economic orbit use. We connect our results to the literature on natural resource exploitation and conclude in section \ref{sec:discussion}. Proofs, derivations, and extensions are shown in the Appendix.

\section{A simple model of orbit use}
\label{sec:simple_econ}

We begin by presenting a stylized model of satellite launch and operation with the goal of isolating the interplay between the economic and physical mechanisms in orbital use. Time is finite and divided into three periods, $t=0,s,l$. In period $0$ launches take place; in period $s$ and $l$ the satellites operate and produce revenue. Between periods, satellites face the risk of collision, which crucially depends on the total number of satellites in orbit. 

\paragraph{Orbit:} Consider a spherical shell (``orbit'') around the Earth be divided into a finite number of slots, each slot able to hold one satellite. Once in orbit, due to unpredictable perturbations, satellites can collide with each other or with orbiting debris. Upon collision, the satellites become inoperable but remain in orbit in the form of debris. 

Let $X\in\mathbb{R}$ denote the objects in orbit at a given period, whether operating satellites or debris.\footnote{In our analysis, we treat each individual satellite launch as an infinitesimal increment to $X$. None of our results depend critically on this, while their mathematical characterization is substantially simplified. } The probability of a satellite surviving into the next period depends negatively on the total objects in orbit and is denoted by $q(X)\in[0,1]$, with $q'(X)\leq0$. We assume that there exists a $\bar{X}<\infty$ such that $q(X)=0$ for all $X\geq\bar{X}$. Intuitively, $\bar{X}$ corresponds to the number of orbiting objects at which the orbit becomes unusable because any satellite would not survive after launch with probability 1. 

We assume that launches can take place only in period $0$, and the stock of objects in orbit at the end of period $0$ always corresponds to the satellite launched at time 0, which we denote by $S\in\mathbb{R}$.\footnote{The assumption on launches is made for simplicity. In the full dynamic model, launches can happen in any period. } The probability of a satellite operating in the period $s$ corresponds to its probability of not colliding with any objects, so $q(S) S$ are the satellites that operate in the short run. When a satellite collision occurs, either with another satellite or debris, fragments are produced so that more than one individually-orbiting objects result. We assume that fragments in excess of $1$ form at the rate $\sigma > 0$.\footnote{For simplicity, in this setting we do not explicitly model debris-debris collisions or natural decay of debris from orbit, though we incorporate both features separately in our full model.} The satellites that collide between periods 0 and $s$ are $(1-q(S))S$, and they generate $(1+\sigma) (1 - q(S))S$ units of debris. The number of orbiting objects at the end of period $s$, which we denote by $g(S)$, is thus the sum of the satellites operating in period $s$ and the debris formed from the collisions. Summing the two and rearranging one obtains
\begin{align}
    g(S) \equiv S + \sigma (1 - q(S))S.\label{eqn:LOMS}
\end{align}
The formation of fragments from collision, i.e. $\sigma>0$, implies that the total objects in orbit increase between the end of period $0$ and the end of period $s$, that is $g(S)>S$. Assuming that collision probabilities are independently distributed, the probability that a satellite operates in period $l$, evaluated at the time of its launch, is thus $q(S)q(g(S))$.

\paragraph{Economic returns:} The launch of a satellite requires a unitary cost, $F>0$, to be incurred at time 0. The satellite, conditional on being operational, generates a constant per-period return, $\pi>0$, in periods $t=s,l$, discounted at the per-period rate $r$. The model can be generalized to allow for the return $\pi$ to be a function of $S$, and we do so in Appendix \ref{appendix:downward-demand} for the simple model, and in the simulation exercise in Section \ref{sec:simulation} for the full dynamic model. While the dependence of $\pi$ on $S$ has implications for the quantitative simulation, the qualitative features of our results do not critically depend on it.

Let $V(S)$ denote the expected present value of the lifetime net return from a satellite when $S$ satellites are launched. This can be written as
\begin{align}
    V(S) &= -F + \frac{\pi q(S)}{1+r} +  \frac{\pi q(S)q\left(g(S)\right)}{(1+r)^2}.\label{eqn:EXP}
\end{align}
The structure of the discounted expected return at time 0 in equation \eqref{eqn:EXP} is intuitive: the expected return in period $s$ is $\pi q(S)$, while in the long run it is $\pi q(S)q(g(S))$, both appropriately discounted. The marginal expected profits of an additional satellite are decreasing in the total number of satellites launched. Intuitively, an additional satellite launched increases the number of objects in orbit in both the short-run and the long-run, hence decreasing the survival probability in both periods. Additionally, equation \eqref{eqn:LOMS} suggests that the long-run negative impact of an additional satellite on the survival probability is increasing in $\sigma$.

\paragraph{Equilibrium with open access:} We consider first the equilibrium choice of satellite launches, $S$, under open access. It is assumed that slots cannot be assigned property rights---satellite operators own the satellites they launch, but not the slots they occupy. Lacking slot rights, operators cannot exclude each other from taking unoccupied slots despite the collision risks posed. 

Under open access, launches take place until the expected profits are equal to zero, i.e. $V(S)=0$. Using equation \eqref{eqn:EXP} with equation \eqref{eqn:LOMS}, the open-access equilibrium launches, $\hat{S}$, correspond to the level at which the private marginal cost of a launch is equal to the expected marginal revenue:
\begin{equation}
   V(\hat{S})=0 \implies \underbrace{F}_{\text{\shortstack{Private cost\\of a satellite\\(PC)}}} = ~~~~~
   \underbrace{\pi \frac{q(\hat{S})}{1+r}\left( 1 + \frac{q(g(\hat{S}) )}{1+r}\right)}_{\text{\shortstack{Expected revenue\\from a satellite\\(ER)}}}. \label{eqn:OPA}
\end{equation}
Since $q(0)=1$ and $g(0)=0$, an open-access equilibrium with positive launches, i.e. $\hat{S}>0$, occurs whenever the cost of launching is smaller than the expected return of the first satellite launched, i.e. $F \leq \pi \frac{1}{1+r}\left( 1 + \frac{1}{1+r}\right)$. We assume this condition holds. Note that in choosing whether to launch, an individual agent takes $S$ as given, and thus does not account for the impact that their additional launch will have on the probability of survival of the entire satellite fleet. 

Figure \ref{fig:simple-model-mbmc} shows the diagram representation of the economic problem of a satellite operator. The $PC$ curve represents the private cost of launching, which is constant in our model and equal to $F$. The $ER$ curve represents the expected lifetime return of one satellite, which corresponds to the right-hand side of \eqref{eqn:OPA}. The expected return is decreasing in the number of satellites in orbit, $S$, due to the lower survival probability when $S$ is larger. The difference in the $ER$ curves between panels A and B reflects the higher $\sigma$ ($1.25$ vs $20$): for any level $S$, the increase in debris formation increases the objects in orbit, and thus reduces the survival probability between the short and the long run, shifting the $ER$ curve downward. The point $OA$ indicates in both panels the equilibrium $\hat{S}$ that satisfies equation \eqref{eqn:OPA}. As expected, the open-access equilibrium launches are lower when $\sigma$ is larger.

\begin{figure}[t] 
	\centering
	\includegraphics[width=0.8\textwidth]{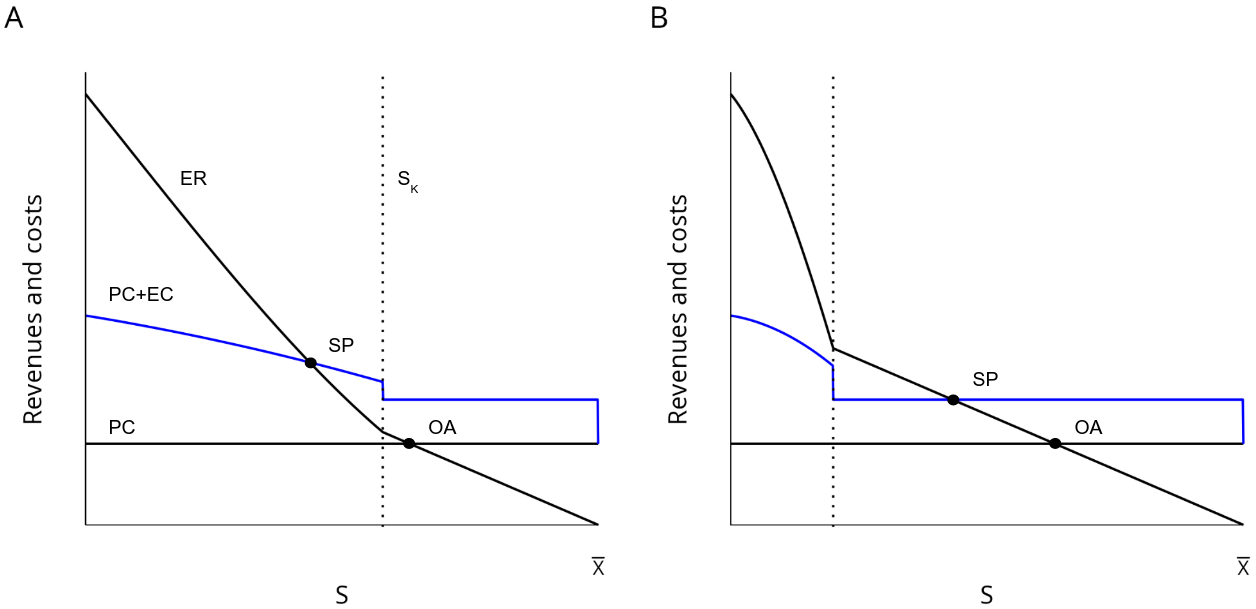} 
	\caption{\footnotesize  Equilibrium diagram for the simple model. Curves are labeled as in equations \ref{eqn:OPA} and \ref{eqn:PLN}. $ER$ shows the expected revenue of one additional satellite as a function of $S$, $PC$ shows the private cost of one additional satellite, and the dashed curve $PC + EC$ shows the social cost of one additional satellite (i.e. the sum of private and external costs). The survival probability is specified as $q(X) = 1- X/\bar{X}$. The parameters are set to $\pi = 1, r = 0.05, F = 0.35, \bar{X} = 5$, with $\sigma = 1.25$ in panel A and $\sigma = 20$ in panel B. The dotted vertical lines show the Kessler threshold. Under the chosen parameterization, open access causes Kessler Syndrome in both panels A and B, while the social planner only causes Kessler Syndrome in panel B. 
	}
	\label{fig:simple-model-mbmc}
\end{figure}

\paragraph{Social planner:} Now consider the case when the orbit is operated by a single agent, which we name the social planner. We assume that the objective of the planner is choose a launch level $S^*$ so as to maximize the expected net present value of the satellite fleet, $S V(S)$. This requires the first-order condition
\begin{align}
    S^*:V(S^*) = \underbrace{-\frac{\pi}{1+r}\left[q'(S^*)\left(1+\frac{q(g(S^*))}{1+r}\right) + q(S^*)\frac{q'(g(S^*))g'(S^*)}{1+r}\right]}_{\text{External cost of a satellite (EC)}}.\label{eqn:PLN}
\end{align}
Compared to the open-access condition, $V(\hat{S})=0$, equation \eqref{eqn:PLN} shows that the planner internalizes the effect of an additional satellite (represented here by the infinitesimal change $dS$) by considering its impact on the collision probabilities faced by the rest of the fleet, which is the external cost term $EC$ in the equation. The external cost has two parts: the first represents the negative impact on the survival probability in the short-run (recall $q'(S)<0$), the second represents the negative impact on the survival probability between the short-run and the long-run, captured by the term $q'(g(S))g'(S)$. In particular, the planner correctly anticipates that a collision in the short run generates debris (represented by the factor $g'(S)>0$), which impacts the probability of survival between the short and the long run, when $\sigma>0$. The simple model offers an important insight: the external cost of a satellite consists of its impact on survival probabilities from that point onward, \textit{taking into account the physical dynamics of objects in orbit,} summarized here by the law of motion $g(S)$. 
Since $V(S)$ is decreasing in $S$, and the term in square brackets in equation \eqref{eqn:PLN} is always negative, it follows that $S^*<\hat{S}$. The diagram in Figure \ref{fig:simple-model-mbmc} shows in blue lines the total social cost of an additional satellite, the $PC+EC$ curve. The curve is always above the $PC$ curve because $EC>0$ at all times, and it is downwards sloping because the external cost of a satellite declines as the orbit saturates.\footnote{The discountinuity of the $PC+EC$ curve is the consequence the specification we have chosen for $q(X)$, which is not differentiable at $\bar{X}$.} The planner's solution is denoted by the intersection $SP$ which corresponds to a lower number of satellites compared to the open access case. By internalizing the externality imposed by each satellite on the probability of collision, the planner chooses a smaller fleet size than open access. 

\paragraph{Kessler Syndrome:}
A central question in our analysis is how the consideration of economic incentives interacts with the physics of orbiting objects in determining the level of orbit congestion. Specifically, we focus on a particularly undesirable congestion outcome known as Kessler Syndrome. Kessler Syndrome is defined as a situation in which the current level of objects in orbit implies that the orbit will eventually become unusable in finite time. This level is sometimes referred to as the ``runaway growth threshold.'' In the context of our model, it corresponds to a level of satellite launches that keep the orbit usable in the short run, i.e. $S<\bar{X}$, while also making the orbit unusable in the long run, i.e. $g(S)\ge\bar{X}$. The idea can be formalized by considering a threshold of objects in orbit in the short run, $S_K$, such that for any level above the threshold, there is a zero probability of satellite survival in the long run. Using equation \eqref{eqn:LOMS}, the threshold is thus defined by
\begin{equation}
S_K:g(S_K)= \bar{X}.\label{eqn:KESS}
\end{equation}
An important question that we consider is whether it is possible for Kessler Syndrome to occur in the open-access equilibrium or in the social planner's allocation. Towards answering that question, we first note that in our simple model it must be that $S_K<\bar{X}=g(S_K)$ for Kessler Syndrome to be relevant at all. That is, the critical Kessler threshold must be below the level of objects in orbit that make the orbit immediately unusable. If not, Kessler Syndrome could never occur as an equilibrium or optimum outcome: since satellites do not earn returns until the short run, rendering the orbit unusable in that period guarantees that the fixed cost of launching will not be recovered. A further implication of this argument is that $S_K<g(S_K)$, which requires a positive net debris generation rate $\sigma>0$. Intuitively, a runaway growth dynamic can happen only when a collision results in more debris objects than the number of objects involved in the collision, i.e. $\sigma > 0$. 
To see this, suppose that we specify, $q(X) = 1- X/\bar{X}$, then $g(S) = S\left(1+ \sigma S/\bar{X} \right)$, and using \eqref{eqn:KESS} one can show that $S_K = \left(\frac{\sqrt{1+4\sigma}-1}{2\sigma}\right)\bar{X} .$ Note that $S_K<\bar{X}$ for $\sigma>0$, and $\lim_{\sigma\rightarrow 0} S_K = \bar{X}$, so that Kessler Syndrome would only be possible when $\sigma>0$. A more general version of this condition applies in the more-realistic setting of our full model, where debris result from collisions and also decay from orbit over time. 

The Kessler Syndrome threshold $S_K$ depends on the physics of orbiting objects, summarized here by $q(X)$ and $\sigma$. However, whether the threshold is crossed as a consequence of launching choices depends on the economics of orbit use. The following Proposition traces the connection between the physics and the economics that is at the center of our analysis.

\begin{restatable}[Kessler Syndrome]{prop}{simpleKessler}
\label{prop:simple-kessler}
Let the dynamic model of objects in orbit be characterized by equation \eqref{eqn:LOMS} with $\sigma > 0$, and the Kessler threshold $S_K$ be defined as in equation \eqref{eqn:KESS}, then
\begin{enumerate}
\item under the open-access equilibrium, Kessler Syndrome occurs if
\begin{equation}
    \frac{\pi}{1+r}q(S_K) \geq F; \label{eqn:KS_OA}
\end{equation}
\item under the social planner allocation, Kessler Syndrome occurs if
\begin{equation}
    \frac{\pi}{1+r}\left[q(S_K) + S_K q'(S_K)\right] \geq F. \label{eqn:KS_PL}
\end{equation}
\end{enumerate}
\end{restatable}

\noindent The proof is shown in the Appendix. Condition \eqref{eqn:KS_OA} states that under open access, equilibrium launches at time $0$ exceed $S_k$ and trigger the Kessler Syndrome in between the short and the long run when the expected short-run return is higher than the cost of launching. At the individual level, it is optimal to launch because expected profits are positive, even if just accruing in the short-run. Condition \eqref{eqn:KS_PL} states that for the social planner, the expected return of the additional satellite launched at $0$ has to be corrected for the external cost imposed on the probability of survival of all the existing satellites, which is represented by the negative term $S_K q'(S_K)$. Proposition \ref{prop:simple-kessler} thus shows that Kessler Syndrome is possible under both open access and the social planner, but under more stringent conditions for the latter. The diagram in Figure \ref{fig:simple-model-mbmc} shows the Kessler threshold $S_K$ as the dotted vertical line in both panels. In panel A, the open-access equilibrium is above the threshold, while the social planner solution is not. In panel B, both open access and social planner are above the threshold. Comparing the two panels shows how the physics and the economics interact: $S_K$ is lower in panel B because of the impact of a larger $\sigma $ on the debris growth, but, at the same time, the expected revenue and the expected social cost curves are lower, resulting in both outcomes to be above the threshold.

Using again the functional form $q(X) = 1- X/\bar{X}$, the condition for Kessler Syndrome under open access is $F\leq\frac{\pi}{1+r}\frac{1+ 2\sigma-\sqrt{1+4\sigma}}{2\sigma}$, while the condition under the social planner is $F\leq\frac{\pi}{1+r}\frac{1+ \sigma-\sqrt{1+4\sigma}}{\sigma}$. Both conditions are more likely to be satisfied if more debris result from a collision (larger $\sigma$), the launch cost is lower (smaller $F$), or if the present value of profits is larger (larger $\frac{\pi}{1+r}$). 

In summary, the simple model offers two important insights with respect to the Kessler Syndrome. First, Kessler Syndrome requires autocatalytic debris growth---represented by $\sigma>0$---that operates in addition to the choice of how many objects to send in orbit. The reason is that launching intensity endogenously adjusts to the evolution of the collision risk. This is a key distinctive feature of our model compared to physical models with an exogenous launching rate. Second, Kessler Syndrome can happen both under open access and the social planner, which means that the full internalization of the external cost of a satellite is not enough to prevent the orbit from becoming unusable. We next develop a fully dynamic model with a more sophisticated physical and economic structure that allows us to expand on the insights offered by the simple model and formulate quantitative predictions.

\section{Physics of orbit use}
\label{sec:physics_only}

In this section we develop a more realistic physical model of orbit use. Consider a spherical shell around the Earth, say the region between 600-650 km above mean sea level.\footnote{This ``shell of interest'' approach is frequently used in debris modeling, e.g. \cite{rossi1998} and \cite{bradleywein2009}, though higher fidelity models use large numbers of small regions to track individual objects, e.g \cite{liouetal2004}, \cite{lioujohnson2008_soi}, and \cite{lioujohnson2009_2}. Paths which span shells are possible and useful for some applications, but highly elliptical orbits (such as Molniya orbits) are the exception rather than the rule. A Molniya orbit has a low perigee over the Southern Hemisphere and a high apogee over the Northern Hemisphere. Molniya orbits require less power to cover regions in the Northern Hemisphere (e.g. former Soviet Union countries) than geosynchronous orbits, due to the low incidence angles of rays from the Northern Hemisphere to geosynchronous positions. } There are $S_t$ satellites and $D_t$ debris fragments in the shell in period $t$, and each satellite collides with another object (satellites or debris) with probability $L(S_t,D_t)$.\footnote{It is more convenient to parameterize the physical model in terms of the probability of a collision, $L$, rather than the probability of no-collision, $q$, as we have done in the simple model in Section \ref{sec:simple_econ}. It is of course the case that $q=1-L$, and we will use this relationship later in the paper. }
The laws of motion for satellites and debris are specified as:
\begin{align}
\label{eqn:satLoM}
S_{t+1} &= S_t(1-L(S_t,D_t)) + X_t \\
\label{eqn:debLoM}
D_{t+1} &= D_t(1-\delta) + G(S_t,D_t) + mX_t .
\end{align}
The number of active satellites in orbit in period $t+1$, $S_{t+1}$, is equal to the number of launches in the previous period, $X_t$, plus the number of satellites which survived the previous period, $S_t(1-L(S_t,D_t))$. Since we assume active satellites are identical, $L(S_t,D_t)$ represents the proportion of orbiting active satellites which are destroyed in collisions. We assume that the collision probability is twice continuously differentiable, nonnegative, increasing in each argument, and bounded below by 0 and above by 1.\footnote{Satellite operators try to avoid collisions by maneuvering their satellites when possible; the collision probability in this model should be thought of as the probability of collisions which could not be avoided, with easily avoided collisions optimized away. Collisions which could have been avoided but were not due to human error are included in $L$. Implicitly we are assuming operators are imperfect cost-minimizers when maneuvering satellites. Even when operators can maneuver their satellites and are aware of impending collisions, coordination can be challenging and plagued by technical glitches \citep{guardianSpaceXESA}.} 

The amount of debris in orbit in $t+1$, $D_{t+1}$, is the amount from the previous period minus the proportion of debris that decay at rate $\delta>0$, which gives $D_t(1-\delta)$, plus the number of new fragments created in collisions, $G(S_t,D_t)$, plus the amount of debris in the shell created by new launches, $mX_t$. Note that the function $G(S_t,D_t)$ plays the role of the debris formation parameter $\sigma$ in the simple model.\footnote{Satellites eventually cease being productive and their orbits eventually decay without fuel expenditure to maintain their path, features we abstract from in the main model. This abstraction makes some results a little clearer but does not qualitatively affect them. \cite{rouillon2019} includes these features in the main model and derives qualitatively similar results, albeit in a continuous-time setting. We include limited lifetimes, time-varying payoffs and costs, and state-dependent satellite payoffs in our calibrated simulations in Section \ref{sec:simulation}.}  We assume the new-fragment function $G$ is twice continuously differentiable, nonnegative, strictly increasing in each argument, zero when there are no objects in orbit ($G(0,0)=0$), and unbounded above ($\lim_{S \to \infty} G(S,D) = \lim_{D \to \infty} G(S,D) = \infty$).\footnote{While unbounded debris growth is unphysical without an unbounded influx of mass, this is a standard modeling choice in the debris modeling literature focused on the next few decades and centuries. We provide an argument for this choice over these timescales based on the kinetic energy-object size scaling law and empirical data on fragment sizes following collisions between large intact objects in Appendix \ref{appendix:debris_unbounded}.} We also assume the effect of the first satellite or debris fragment on new fragment formation is negligible, i.e. $G_S(0,D) = G_D(S,0) = 0$ (letting subscripts denote partial derivatives). 

For simulations and figures, we use the following functional forms:
\begin{align}
    L(S,D) &= 1 - e^{-\alpha_{SS}S - \alpha_{SD} D} \label{eqn:general-collision-function-example} \\
    G(S,D) &= \beta_{SS} (1 - e^{-\alpha_{SS} S}) S + \beta_{SD} (1 - e^{-\alpha_{SD} D}) S + \beta_{DD} (1 - e^{-\alpha_{DD} D}) D,
    \label{eqn:general-new-debris-function-example}
\end{align}
where $\alpha_{SS}, \alpha_{SD}, \alpha_{DD}, \beta_{SS}, \beta_{SD}, \beta_{DD}$ are all positive physical parameters.  We derive these forms and give some physical intuition about the parameters in Appendix \ref{appendix:kinetic_gas_model}, and describe physical model calibration in Appendix \ref{appendix:physical_model_calibration}. Calibrations and scales for figures illustrating qualitative results are chosen to emphasize relevant features.\footnote{For figures illustrating qualitative results, $S$ is measured in terms of the number of satellites necessary for the collision probability in a zero-debris environment to be $L(S,0) = 1 - e^{-\alpha_{SS}S}$. Similarly, $D$ is measured in terms of the number of debris fragments in orbit required for $G(0, D) = \beta_{DD} (1 - e^{-\alpha_{DD} D}) D$ new fragments to form. For example, suppose $\alpha_{SS} = 10^{-4}$ and 1 unit of satellites produces a collision probability of 0.01. Then each unit of satellites corresponds to roughly 100 satellites. Similarly, suppose $\beta_{DD} = 0.9$, $\alpha_{DD} = 10^{-5}$, and 1 unit of debris produces 10 new fragments. Then each unit of debris corresponds to roughly 1000 fragments. The launch rate is measured in units of satellites, and inherits its scale from $S$. The values of $\alpha_{SS}, \alpha_{DD}, \beta_{DD}$ used for these rescalings are obtained from the calibration to actual data described in Appendix \ref{appendix:physical_model_calibration}. We scale figures illustrating qualitative results so that $L(1,0) = \frac{\pi}{F} - r$ and so that $G(0,1) = 0.1$.} The functional forms do not affect our qualitative results.

\section{Economics of orbit use}
\label{sec:full-econ-model}
We now present an economic model whose elements match the physical setting described by equations \eqref{eqn:satLoM} and \eqref{eqn:debLoM}. The key focus is in modeling the economic choice behind the launch rate, $X_t$. The launch rate will be determined by forward-looking optimizing agents based on the economic institutions (open access or a social planner) and the current state of the environment (i.e. $(S_t,D_t)$). 

We assume satellites are identical and infinitely-lived unless destroyed in a collision. As in the simple model, we assume that the net per-period payoff for a satellite is $\pi > 0$, while it costs $F>0$ to plan, build, and launch.\footnote{In reality, satellites accrue one-time (capex) and recurring (opex) costs, while the revenues are recurring annual flows. Satellite capex costs are generally significant. \citet{ogutu2021techno} estimate capex and opex costs for several proposed satellite constellation designs, finding that capex costs are roughly 2-4x the size of annual opex costs. For theoretical analysis, it is convenient to write the capex costs separately in $F$ while packaging the recurring revenues and costs into $\pi$.} Operators discount the future at the common rate $r>0$. Costs and payoffs are assumed to be constant over time and states throughout the analysis. None of our results depend qualitatively on these assumptions. We relax these assumptions in the calibrated simulations in Section \ref{sec:simulation}, allowing for finite satellite lifetimes, time-varying costs and payoffs, and payoffs which depend on the current stock of active satellites in the shell (e.g. reflecting downward-sloping demand for satellite services). 

Finally, we will maintain the assumption that a launch rate cannot be negative, that is $X_t\geq0$. This reflects the fact that there are currently no technologies which can remove a satellite from orbit other than a decision by the operator to deorbit their own satellite.\footnote{Even these decisions may not be actionable, e.g. satellites without propulsion cannot be deorbited on command and must instead wait for natural decay processes to remove them from orbit.} As long as it was optimal to launch a satellite, it will never be optimal for a private operator to deorbit the satellite before its useful life is over. This makes satellites ``putty-clay'' investments \citep{johansen1959substitution}: once the satellite is in orbit, it will be operated until it is destroyed in a collision. This generates launch rate policy functions which have an ``inaction region'' where the launch rate is at its lower bound, i.e. states $(S,D)$ such that $X = 0$.

\subsection{Open access}

\subsubsection{Equilibrium}
We assume that all economic agents are risk neutral so they are concerned only with expected discounted values. The value at time $t$ of a satellite in orbit, which we denote by $Q(S_t, D_t)$, is thus the sum of present payoffs and the expected discounted value of its remaining lifetime payoffs, and it can be represented recursively by:
\begin{equation}
Q(S_t,D_t) = \pi + \frac{1}{1+r} (1-L(S_t,D_t)) Q(S_{t+1},D_{t+1}).
\label{eqn:satBellman}
\end{equation}
A firm which does not own a satellite in period $t$ decides whether to pay $F$ to plan, build, and launch a satellite which will reach orbit and start generating payoffs in period $t+1$, or to wait and decide again in period $t+1$. There is a large number of potential launchers, each one denoted by $i$. At each period, the launcher $i$ has to decide whether to launch, corresponding to choosing $x_{it}=1$, or to wait for next period, corresponding to choosing $x_{it}=0$. If launching, the future value for $i$ will be equal to \eqref{eqn:satBellman} net of the one time cost $F$. If not launching, the future value will correspond to the value of a potential launcher at $t+1$. The problem can be written recursively as
\begin{align}
V_{i}(S_t,D_t,X_t) &= \max_{x_{it} \in \{0,1\}} \left\{ \frac{(1-x_{it})}{1+r} V_{i}(S_{t+1},D_{t+1},X_{t+1}) + x_{it} \left[ \frac{1}{1+r} Q(S_{t+1},D_{t+1}) - F \right]  \right\} ,
\end{align}
subject to the laws of motion \eqref{eqn:satLoM} and \eqref{eqn:debLoM}. The launch rate at time $t$ is then given by, $X_t = \int_i x_{it}di$.

In an open-access equilibrium, the launch rate is such that all launching firms earn zero profits, i.e. the equilibrium launch rate $\hat{X}_t$ solves
\begin{equation}
    V_i(S_t, D_t, \hat{X}_t) = 0. \label{eqn:launcher_zero_profit}
\end{equation}
Combining equations \eqref{eqn:satBellman}-\eqref{eqn:launcher_zero_profit}, one obtains that in an open-access equilibrium it must be that
\begin{equation}
    F = \pi\left(\frac{1}{1+r - q(S_t,D_t,\hat{X}_t)}\right), \label{eqn:zeroprofit}
\end{equation}
where $q$ represents the survival probability of a satellite from $t$ to $t+1$,
\begin{equation}
q(S_t,D_t,\hat{X}_t) \equiv 1 - L(S_{t+1}(S_t,D_t,\hat{X}_t),D_{t+1}(S_t,D_t,\hat{X}_t)).
\end{equation}
Equation \eqref{eqn:zeroprofit} is the infinite-horizon version of the open-access condition \eqref{eqn:OPA} in the simple model. It states that under open access, the equilibrium launch rate equates the cost of launching, $F$, to the per-period payoff $\pi$, discounted by the rate $r$ adjusted by the survival probability of the satellite into next period, $q(S_t,D_t,\hat{X}_t)$. Intuitively, when the survival probability is higher (lower), the effective discount rate $r - q(S_t,D_t,X_t)$ is lower (higher), and thus more (fewer) launches would take place. Note that the survival probability $q(S_t,D_t,\hat{X}_t)$ embeds the physics of orbit use, which impacts the economics of orbit use by adjusting the effective opportunity cost of investing in the satellite asset via \eqref{eqn:zeroprofit}. At the same time, the choice of investing, represented by $\hat{X}_t$, impacts the physics of orbit use via the laws of motion \eqref{eqn:satLoM} and \eqref{eqn:debLoM}, which creates the fundamental feedback between the physics and the economics of orbit use.

\subsubsection{Dynamic analysis of open-access orbit use}

In this section we explore the dynamics of orbit use under open access, with the goal of highlighting the impact of the economic decision of satellite launchers on the evolution of objects in orbit. We first note that equation \eqref{eqn:zeroprofit} can be rewritten as
\begin{equation}
   \frac{\pi}{F} - r = L(S_{t+1},D_{t+1}). \label{eqn:zeroprofit-rate}
\end{equation}
which shows that the open-access equilibrium is on an isoquant of the collision probability function---specifically, the isoquant at the level of the ``excess rate of return,'' $\pi/F - r$. Changes in returns, costs, or discounting which lead to the same changes in excess return will have identical effects on open-access behavior. The open-access system is a solution in $(X_t, S_{t+1}, D_{t+1})$ to equations \eqref{eqn:satLoM}, \eqref{eqn:debLoM}, \eqref{eqn:zeroprofit-rate}.

Figure \ref{fig:example-trajectories} shows two examples of open-access trajectories for the launch rate, satellite and debris stocks, and collision probability, under different excess rates of return. In both cases, the launch rate is initially high to ensure the equilibrium condition holds (i.e. to ensure that the collision probability is on the required isoquant). In subsequent periods, the physical dynamics continue to generate debris, pushing the collision probability above the equilibrium level and continually reducing the satellite stock through collisions. The high collision probability forces the launch rate to be zero until the debris and satellite stocks are low enough that a positive launch rate can satisfy the equilibrium condition. 

Figure \ref{fig:oa-policy-phase} shows an example of the open-access launch policy $\hat{X}(S,D)$ and the corresponding phase diagram. Unlike typical problems in economic dynamics, this problem does not feature a saddle path to the equilibrium or optimum. Instead, there is a basin of attraction to a stable steady state (the ``stable basin''). An initial condition will reach the stable steady state if and only if it is within the stable basin. The size and shape of the stable basin is determined both by the system's physics and the economic institution governing launch rates. 

\begin{figure}[H] 
	\centering
	\includegraphics[width=0.9\textwidth]{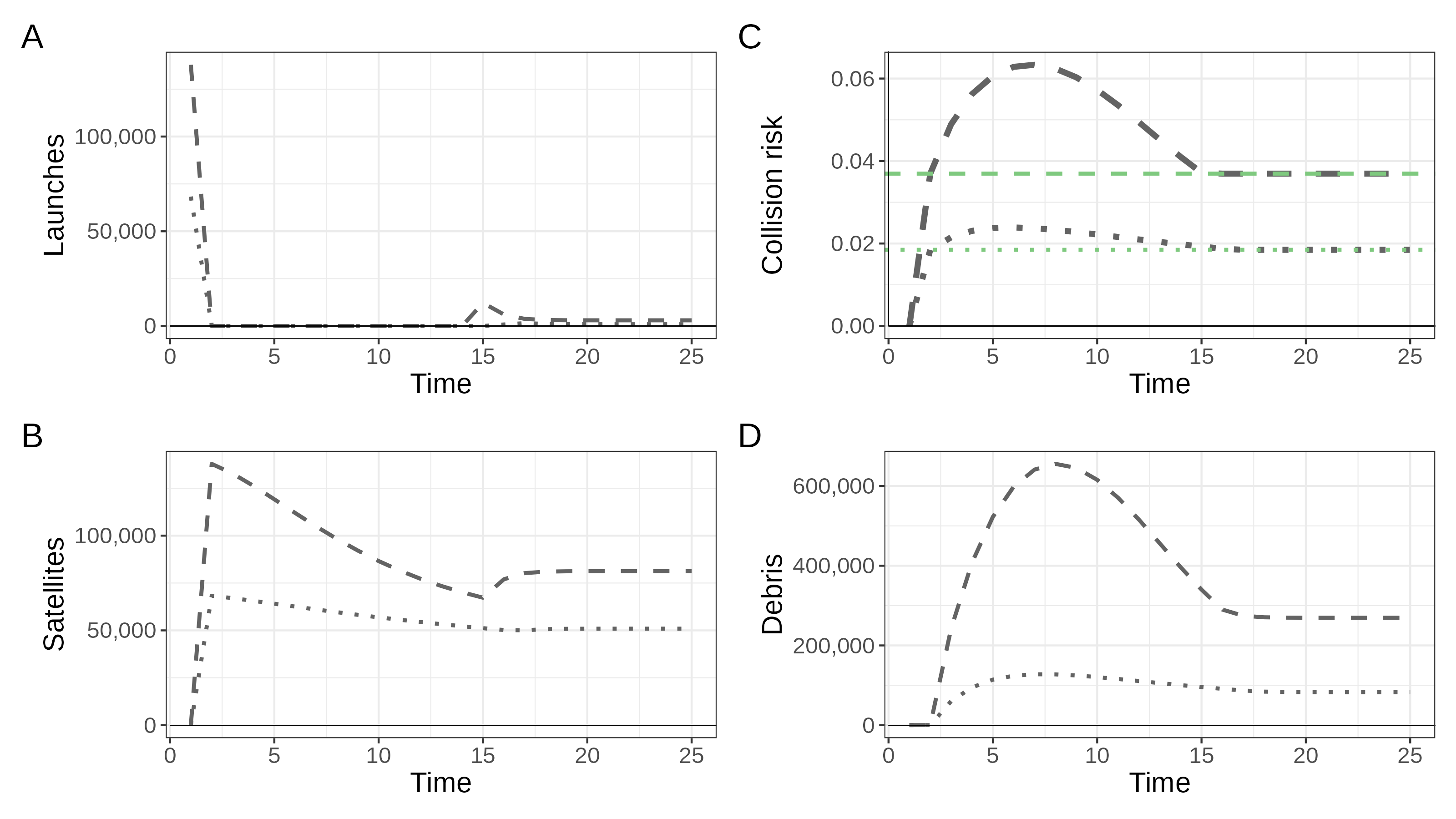} 
	\caption{\footnotesize Example open-access trajectories. The dotted line shows a trajectory under a low excess rate of return (green dotted line in panel C), while the dashed line shows a trajectory under a high excess rate of return (green dashed line in panel C). The initial condition for both trajectories is $(S_0,D_0) = (0,0)$. 
	}
	\label{fig:example-trajectories}
\end{figure}

\begin{figure}[H] 
	\centering
	\includegraphics[width=0.9\textwidth]{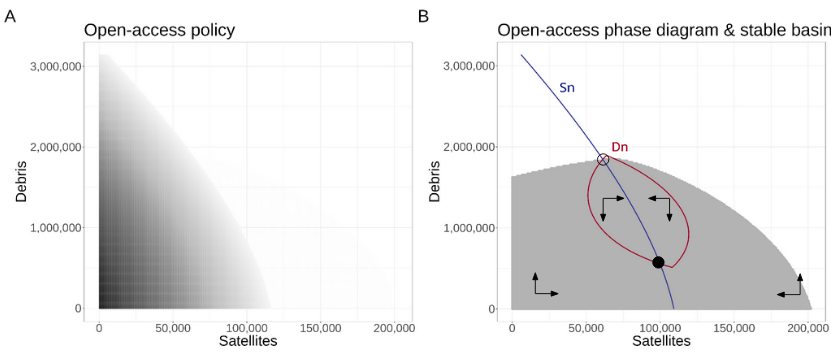} 
	\caption{\footnotesize Policy functions, phase diagrams, and stable basin under open access. In panel A, darker colors correspond to higher launch rates. In panel B the red line shows the nullcline for debris, the blue line shows the nullcline for satellites, and the gray area shows the stable basin. The arrows show the direction of motion, the solid black circle shows the stable steady state, and the open black circle shows the unstable steady state. 
	}
	\label{fig:oa-policy-phase}
\end{figure}

Figure \ref{fig:oa-policy-phase}A shows the open-access launch policy. Launch rates are highest when both satellites and debris are low. As either increases, the launch rate decreases. When the satellite and debris stocks are high enough, the launch rate will be zero. Figure \ref{fig:oa-policy-phase}B shows the phase diagram resulting from the physical dynamics of satellites and debris and the launch policy. The blue line shows the zero-growth isoquant (``nullcline'') for satellites and the red line shows the nullcline for debris. The satellite nullcline takes the shape of the launch policy. As long as $(S,D)$ is to the left of the satellite nullcline, open-access launch behavior will ensure that the stock of satellites is growing. 

The ``kinked ellipse'' shape of the debris nullcline is less obvious. The left side of the kinked ellipse (roughly the portion of the debris nullcline to the left of the satellite nullcline) results when satellites are being launched. In this case debris growth is the sum of launch debris and what physical processes produce. The right side of the kinked ellipse (roughly the portion of the debris nullcline to the right of the satellite nullcline) results when satellites are no longer being launched. In this case debris growth is solely governed by physical processes. Outside the kinked ellipse there is net positive debris growth, either due to satellite launches or collisions between orbiting objects or both. Inside the kinked ellipse the launch rate and physical dynamics are such that the debris stock will shrink even if the satellite stock increases. 

The open-access steady states are pairs $(S,D)$ which satisfy the following conditions:
\begin{align}
L(S,D) &= \frac{\pi}{F} - r, \label{eqm:SSeqmcond} \\
X &= L(S,D)S \label{eqm:SSsateqn}\\ 
\delta D &= G(S,D) + mX. \label{eqm:SSdebeqn}
\end{align}
The intersections of the nullclines in figure \ref{fig:oa-policy-phase}B reveals two steady states satisfying conditions \ref{eqm:SSeqmcond}-\ref{eqm:SSdebeqn}. Proposition \ref{multipleOASS} shows that this multiplicity is a generic property of open-access orbit use whenever (a) active satellites can collide with other active satellites and debris, and (b) debris objects can collide with each other and produce new debris.

\begin{restatable}[Multiplicity and instability]{prop}{multiplicityOASS}
	\label{multipleOASS}
	Given a positive excess return on a satellite and a collision probability function which depends on satellites and debris, multiple open-access steady states can exist if debris objects can collide and produce new debris ($G_D > 0$). An open-access steady state will be stable if and only if
 	\begin{equation}
	    \underbrace{(G_D(S^*,D^*) - \delta)}_{\text{\shortstack{Net rate of\\autocatalytic debris growth}}} < \underbrace{\frac{L_D(S^*,D^*)}{L_S(S^*,D^*)}(G_S(S^*,D^*) + m(\frac{\pi}{F} - r))}_{\text{\shortstack{Rate of new fragment reduction due to\\ equilibrium launch activity response to debris}}}. \label{eqn:stability-condition}
	\end{equation}
    When $G$ is strictly convex in both arguments and two steady states exist, the higher-debris steady state is unstable.
\end{restatable}

\noindent The proof is shown in the Appendix. This multiplicity is possible because open access makes the equilibrium satellite stock a decreasing function of the debris stock \emph{and} autocatalytic debris growth is possible. In a stable steady state, equilibrium reductions in orbit use due to more debris (i.e. the right-hand side of condition \ref{eqn:stability-condition}) outweigh the incremental debris autocatalysis (i.e. the left-hand side of condition \ref{eqn:stability-condition}). In an unstable steady state, the reverse holds. When $G_D \equiv 0$, the stable basin fills the entire space. Condition \eqref{eqn:stability-condition} combined with equation \eqref{eqn:general-new-debris-function-example} show that when launching decisions are endogenous, i.e. respond to economic incentives, the stock of satellites in orbit $S^*$ adjusts and makes it more difficult for \eqref{eqn:stability-condition} to be violated. Even in open access, the private part of the costs of collision risk are internalized, which affects the launching decisions, and eventually the number of satellites in orbit. However, when it comes to debris, once they have formed, there is no endogenous change in behavior that can directly impact the autocatalytic part of their growth, so the economics of orbit use is much less effective in avoiding runaway debris growth. We finally note that Proposition \ref{multipleOASS} is the full model analog to Proposition \ref{prop:simple-kessler}, which required $\sigma>0$ in the simple model for Kessler Syndrome (i.e instability) to emerge.


The number of open-access steady states still depends on the shapes of $L$ and $G$. For example, if $L$ is strictly increasing and $G$ is strictly convex and sufficiently smooth, there can be up to two open-access steady states. Only one will be stable. Without smoothness, there may be more and with more-complex stability properties. Strict convexity of $G$ means putting more objects in orbit strictly increases the expected number of fragments generated by collisions in each period, e.g. due to fragments from one collision propagating and increasing the likelihood of fragment-generating collisions during that period. Smoothness of $G$ means the absence of sharp thresholds in debris formation, e.g. due to critical transitions in debris collision dynamics. The shaded area in figure \ref{fig:oa-policy-phase}B shows the basin of attraction to the stable open-access steady state for our example functional forms (equations \ref{eqn:general-collision-function-example} and \ref{eqn:general-new-debris-function-example}).  

While the possibility of multiple steady states in a non-convex system, some potentially unstable, is a well-understood result in environmental and natural resource economics (e.g. \citet{shallowlake2003, wagener2003, lemoine2016economics}), it is less-recognized in the growing economic literature on orbit use. While \citet{raoetal2020} acknowledge the possibility of an unstable steady state, it is left unmodeled; \citet{aac2015, rouillon2019} similarly focus on stable long-run outcomes; and \citet{aac2018} show the impossibility of meaningful physical instability in a model without the potential for debris-debris collisions, i.e. $G_D > 0$ for some $D$. The non-convexity driving the potential for multiplicity and instability is the potential for debris-debris interactions, which is shown to be theoretically and empirically relevant in the engineering literature on orbit use (e.g. \citet{lewis2020understanding}). 

Finally, it is worth asking if the overshooting seen in figure \ref{fig:example-trajectories}---where collision risk first exceeds the equilibrium level before reaching the steady state---is specific to particular parameter sets or a generic property of orbit use models. If open access equilibria will tend to monotonically approach the stable steady state, then parameter changes which shift the steady state may not cause costly spikes in the debris stock. Such parameter changes include technological advances (e.g. better-shielded satellites), policy guidelines (e.g. encouragement to use frangibolts instead of exploding bolts for booster separation), environmental processes (e.g. sunspot activity which causes $\delta$ to vary), or economic changes (e.g. increases in the excess return on a satellite). Proposition \ref{prop:OAEovershoot} shows that the overshooting in figure \ref{fig:example-trajectories} is a generic property of open-access equilibria: almost all initial conditions will overshoot the stable steady state in satellites or debris or both. Figure \ref{fig:overshoot1} illustrates this result.

\begin{restatable}[Overshooting]{prop}{overshootingOASS}
	\label{prop:OAEovershoot}
	Suppose the new fragment formation function is strictly convex in both arguments and the launch rate constraint does not bind. Except on a set of measure zero, open access paths from initial conditions with positive launch rates will overshoot the stable open-access steady state in at least one state variable.
\end{restatable}

\begin{figure}[t] 
	\centering
	\includegraphics[width=0.9\textwidth]{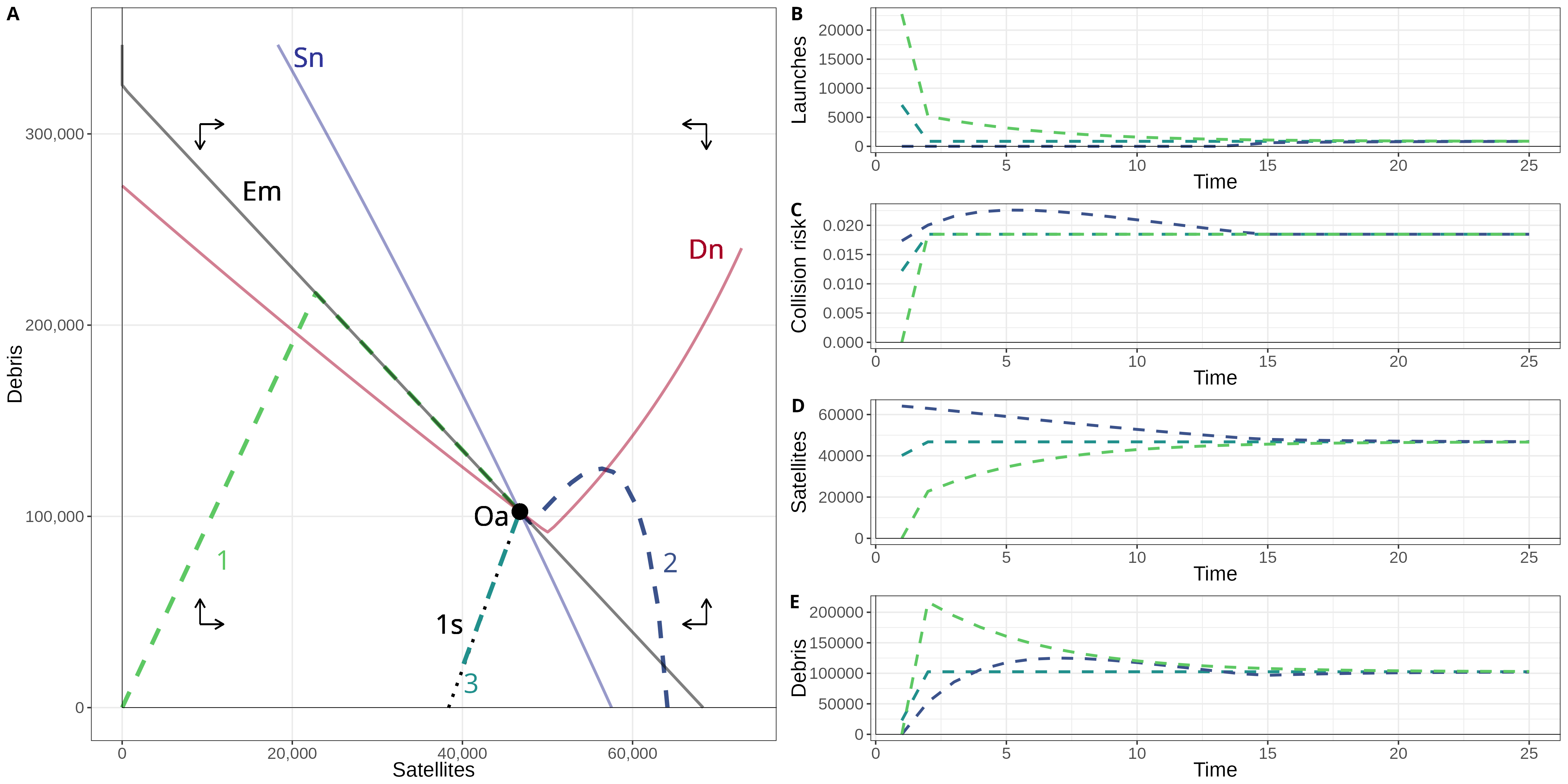} 
	\caption{\footnotesize In panel A, the thin black line shows the equilibrium isoquant (curve \textit{Em}). The thick black point in panel A on the equilibrium isoquant (point \textit{Oa}) is the stable open-access steady state. The dotted black line in panel A connecting the steady state to the y-axis is the manifold of points which can converge to the stable steady state in one step (the ``one-step set'', curve \textit{1s}). Panel A illustrates Proposition \ref{prop:OAEovershoot} with sample paths. The teal dashed line (labeled 3) along the dotted black line shows a path converging to the steady state in one step. The green dashed line (labeled 1) shows a path which overshoots in debris and converges to the steady state along the equilibrium isoquant. The purple dashed line (labeled 2) shows a path which initially overshoots in satellites. These paths follow the equilibrium isoquant to the steady state whenever possible, though the physical dynamics force path 2 to approach the steady state from outside the action region. Panels B-E show the launch rate, collision probability, and satellite and debris stocks for sample paths 1-3 over time. Panel A also shows the satellite (blue, curve \textit{Sn}) and debris nullclines (red, curve \textit{Dn}).}
	\label{fig:overshoot1}
\end{figure}

\noindent The proof is shown in the Appendix. Panel A of Figure \ref{fig:overshoot1} shows three notable features of the open-access equilibrium dynamics. First, curve \textit{Em} is the equilibrium isoquant: a set of points such that the open-access equilibrium condition, equation \eqref{eqn:zeroprofit-rate}, holds.\footnote{The linearity of \textit{Em} is due to the log form assumed for $L(S,D)$.} Paths which eventually converge to the stable steady state (point \textit{Oa}) will attempt to do so along this isoquant. This can be seen most clearly in panel C of Figure \ref{fig:overshoot1}: the green line representing the path labeled 1 reaches the equilibrium collision risk level even as the satellite and debris stocks are away from their steady-state levels. Periods where the collision risk is at the equilibrium level correspond to periods where the path in panel A is moving along \textit{Em}.\footnote{A binding upper bound on the number of satellites launchable per period will cause the path to curve towards the nearest steady state as it approaches \textit{Em}. While this can help firms avoid an unstable steady state following exogenous parameter changes, it does not affect the \emph{existence} of an unstable steady state.} Second, curve \textit{1s} is a frontier separating regions of the state space where, along an equilibrium trajectory, (at least) the stock of debris is guaranteed to exceed the steady-state level at some period from regions where (at least) the stock of satellites is guaranteed to exceed the steady-state level at some period. The frontier demarcated by \textit{1s} is exactly the (measure zero) set of points which can converge to the steady state in one step without overshooting in either state variable. Third, there are in general two types of paths to the stable steady state: those which involve positive launch rates every period (exemplified by paths 1 and 3), and those which involve some periods with zero launch rates (exemplified by path 2).\footnote{Paths where the launch rate is eventually always zero occur if and only if Kessler Syndrome has occurred (see Proposition \ref{prop:finite-horizon-zero-launch} in Appendix \ref{appendix:finite-horizon-orbital-capacity}).} The former occur when the total stock of objects in orbit is low enough, while the latter occur when the total stock is large enough.

\subsection{Social planner}

\subsubsection{Optimum} 
The value of the satellite fleet and the right to launch new satellites is the sum of present and expected discounted payoffs from satellites in orbit and satellites which are going to be launched. The fleet planner chooses the launch rate to maximize this value:
\begin{align}
W(S_t,D_t) =& \max_{X_t \geq 0} \left\{ \pi S_t - F X_t + \frac{1}{1+r} W(S_{t+1},D_{t+1}) \right\}
\label{prog:planner-dp}
\end{align}
subject to the laws of motion \eqref{eqn:satLoM} and \eqref{eqn:debLoM}. The planner's optimal interior launch rate $X^*_t$ equates the flow of benefits and costs from a marginal satellite, such that
\begin{align}
\label{eqn:optflowcond}
F =\pi\left(\frac{1}{1+r-q(S_t,D_t,X^*_t)+\xi(S_t,D_t,X^*_t)}\right) 
\end{align}
where $\xi(S_t,D_t,X^*_t)\geq 0$ is the external cost of a marginal satellite launched in $t$, corresponding to the right-hand side term in the simple model condition \eqref{eqn:PLN}. Equation \eqref{eqn:optflowcond} shows that the planner internalizes the impact of each additional satellite on the survival probability of the satellite fleet by adjusting upward the effective discount rate of per-period returns by the external cost factor $\xi$, compared to the open-access condition \eqref{eqn:zeroprofit}. While the general expression for the external factor is difficult to summarize due to the presence of shadow values and differences between future and past periods, the steady-state expression for an interior launch rate is more interpretable.\footnote{We derive equation \ref{eqn:optflowcond} and $\xi(S_t,D_t,X^*_t)$ in the general (not-necessarily-interior, non-stationary) case in Appendix \ref{appendix:MECderivation}.}  Suppressing function arguments and time subscripts, the external cost of a marginal satellite in an interior steady state is
\begin{align}
\xi =& \underbrace{L_S S }_{\text{\shortstack{``congestion'' channel: \\ marginal cost of satellites colliding \\with each other due to crowding}}} + \underbrace{\frac{1}{1+r} \left(G_S  + m(L + S L_S)\right) L_D S  + \left(1-\frac{1}{1+r}\right) m L_D S}_{\text{\shortstack{``pollution hazard'' channel: \\ marginal cost of satellites colliding\\with new fragments and launch debris}}} \nonumber \\
&+ \underbrace{\frac{1 - \delta + G_D}{1+r} \left( \frac{\pi}{F} - r- (L + L_S S)\right)}_{\text{\shortstack{``pollution persistence'' channel: \\ marginal cost of persistent debris \\and debris growth }}} \label{eqn:ss-mec}
\end{align}
The first term of $\xi$, the congestion channel, represents the cost of additional satellite collision probability due to satellite crowding. This is essentially the same externality identified in open-access fisheries in \citet{gordon1954}. As long as the collision probability is coupled with the satellite stock and new satellites weakly increase the probability of satellite-destroying collisions ($L_S \geq 0$), this term is non-negative. This contemporaneous externality ($L_S S $) could be remedied by coordinating satellites perfectly to avoid each other, e.g. using a common slotting architecture for all satellites \citep{arnas2021definition, lifson2022many}.\footnote{Such coordination is technically feasible, but due to the OST it must avoid implying slot ownership.} In the general model this corresponds to the collision probability function not having $S$ as an argument, i.e. $L(D_t)$ instead of $L(S_t, D_t)$. This form is used in some prior economic literature on orbit use, e.g. \citet{aac2015, aac2018}.\footnote{The form in \citet{rouillon2019} is agnostic on this point, since it focuses on the steady state and treats the collision risk as a ``black box'' determined by the launch rate. Disambiguating these effects requires more structure on the collision probability function.} Since this channel shows only crowding due to active resource users, we refer to it as the ``congestion channel'' following the distinction between open-access, congestion, and pollution externalities in \citet{haveman1973}.  Importantly, note that this channel exists even if debris is irrelevant, i.e. $D \equiv 0$ or $L(S, D) = L(S)$ or both---the remaining terms in equation \eqref{eqn:ss-mec} would disappear.\footnote{This cannot be inferred directly from equation \eqref{eqn:ss-mec} or its more-general form, equation \eqref{eqn:general-mec}, in the appendix, since both are derived assuming debris exists and affects collision risk. It can be seen instead from the first-order conditions planner's lagrangian in equation \eqref{eqn:plan-lagrangian} of Appendix \ref{appendix:MECderivation}. If debris didn't affect collision risk or didn't exist, no terms from the debris law of motion would appear in the first-order necessary, complementary slackness, or transversality conditions.}

Pollution is a different case \citep{haveman1973}. The residual from production which reduces environmental quality is typically borne by entities other than the residual discharger---typically a firm imposing costs on consumers. Debris constitutes a pollution stock---though the entities bearing the costs may be the same as those discharging debris, its temporal persistence means that future users will be affected by present or past users' choices. Open-access rent-dissipation leads firms to ignore these intertemporal effects. Critically, this intertemporal separation means that debris pollution will not constrain equilibrium output via rent dissipation at any given time $t$ \emph{except through its effects on the contemporaneous collision risk that operator expects to face}. These channels are explicitly noted in equation \eqref{eqn:ss-mec} as the ``pollution'' channels. 

The second term, the pollution hazard channel, is the cost of additional collisions with debris. There are two components of this channel, reflecting debris costs incurred over time and immediately. The first component, $\left(G_S  + m(L + S L_S)\right) L_D S $, represents the marginal cost of colliding with fragments generated by collisions involving other satellites (including fragments from collisions between satellites and launch debris). This component reflects the long-run hazard to the fleet created by the marginal satellite, as it may be destroyed in a collision and generate additional fragments. The second component, $m L_D S $, represents the marginal cost of the new satellite's launch debris---a short-run hazard to the fleet. As the discount rate approaches zero only the long-run component matters, while as it approaches infinity only the short-run component matters. As the number of launch debris fragments from the marginal satellite goes to zero, this channel reduces to the discounted marginal cost of collisions with fragments of satellite-debris collisions ($\frac{1}{1+r} G_S L_D S$, a long-run cost) only. As long as having more objects in orbit increases the probability of a collision ($L_D > 0$) and collisions with active satellites or launch debris can produce debris ($G_D > 0$ or $m > 0$), this term is non-negative. 

The third term, the pollution persistence channel, is the cost of debris which does not decay and new fragments produced in collisions between debris objects. If all debris in orbit decayed at the end of each period and could not collide with other debris to create new debris ($\delta = 1$ \& $G_D \equiv 0$), this channel would disappear. The cost of this channel is the forgone discounted payoff from a satellite net of collision risk ($L$), congestion costs ($L_S S$), and the opportunity cost of the funds used to deploy the asset ($r$). As long as the excess rate of return weakly exceeds the collision probability and marginal rate of congestion costs ($\frac{\pi}{F} - r \geq L(S,D) + L_S(S,D)S$), this term is nonnegative. 

\subsubsection{Dynamic analysis of optimal orbit use}

Figure \ref{fig:opt-example-trajectories} shows two examples of optimal trajectories for the launch rate, satellite and debris stocks, and collision probability, under the same excess rates of return used in figure \ref{fig:example-trajectories} and from the same initial condition. We induce the change in excess return by lowering the launch cost---the dependence of the external cost of a marginal satellite on returns, private cost, and discounting parameters breaks the symmetry present under open access. In both examples, the launch rate is initially high to take advantage of the lack of debris in orbit. This leads to some overshooting in the satellite stock only (though less than under open access), and the launch rate falls to zero until the satellite stock is closer to its steady-state level. In subsequent periods the satellite stock approaches its steady-state level from above as debris and collision probability grow monotonically. However, the collision probability remains well below the open-access levels. Unlike open access, the optimal collision probability does not overshoot its steady-state level, since the planner internalizes the dynamic externality of launching too many satellites. The planner does however overshoot the steady-state level of satellites: given the initial condition ($(S,D) = (0,0)$), the planner takes advantage of ``clear skies'' before debris accumulates.

Figure \ref{fig:opt-policy-phase} shows an example of the optimal launch policy $X^*(S,D)$ and the corresponding phase diagram. As before, there is a basin of attraction to a stable steady state. The optimal launch policy has a similar shape as the open-access launch policy, but a weakly lower magnitude in every state. This expands the stable basin and partially cuts off the ``kinked ellipse'' of the debris nullcline. Compared to the open-access policy, the optimal policy maintains fewer satellites and less debris in its stable state. The region outside the planner's stable basin in Figure \ref{fig:opt-policy-phase} corresponds to states where negative launch rates would be required to avert Kessler Syndrome; without access to such technologies, the planner is not able to avert Kessler Syndrome in those states.

The steady-state satellite and debris stocks are both lower than their open-access levels, the satellite stock by around half and the debris stock by around 80\%. As in the open-access case it appears that almost all initial conditions will overshoot the stable steady state in at least satellites or debris or both. Figure \ref{fig:overshoot-opt} illustrates this result.

\begin{figure}[H] 
	\centering
	\includegraphics[width=0.9\textwidth]{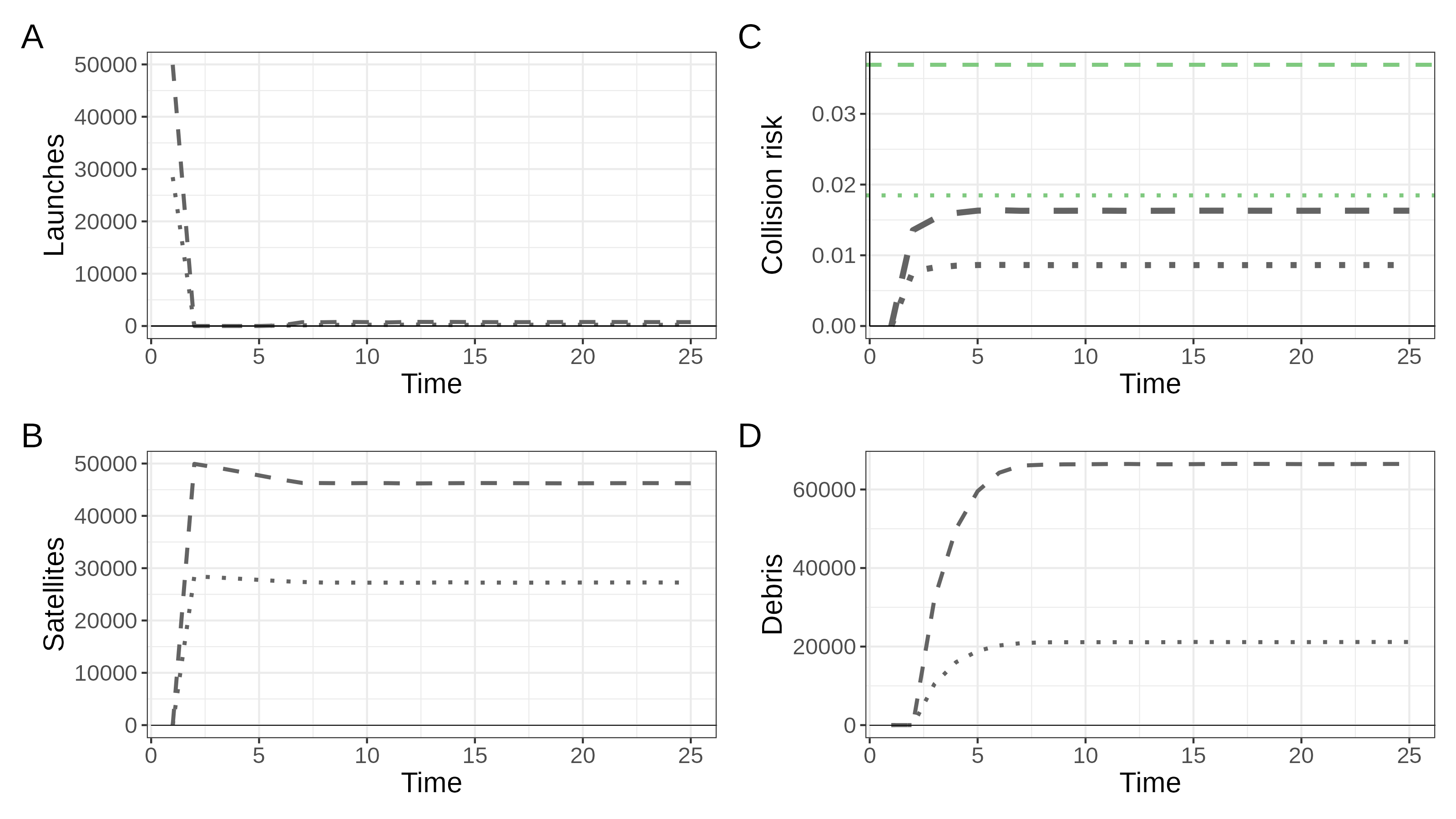} 
	\caption{\footnotesize Example optimal trajectories. The dotted line shows a trajectory under a low excess rate of return (green dotted line in panel C), while the dashed line shows a trajectory under a high excess rate of return (green dashed line in panel C). The initial condition for both trajectories is $(S_0,D_0) = (0,0)$. 
	}
	\label{fig:opt-example-trajectories}
\end{figure}

\begin{figure}[H] 
	\centering
	\includegraphics[width=0.8\textwidth]{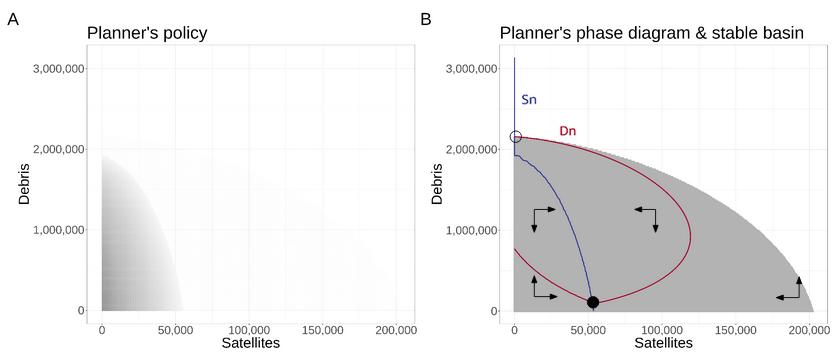} 
	\caption{\footnotesize Policy functions, phase diagrams, and stable basin under the planner. In panel A, darker colors correspond to higher launch rates. In panel B the red line shows the nullcline for debris, the blue line shows the nullcline for satellites, and the gray area shows the stable basin. The arrows show the direction of motion, the solid black circle shows the stable steady state, and the open black circle shows the unstable steady state. 
	}
	\label{fig:opt-policy-phase}
\end{figure}

\begin{figure}[H] 
	\centering
	\includegraphics[width=0.9\textwidth]{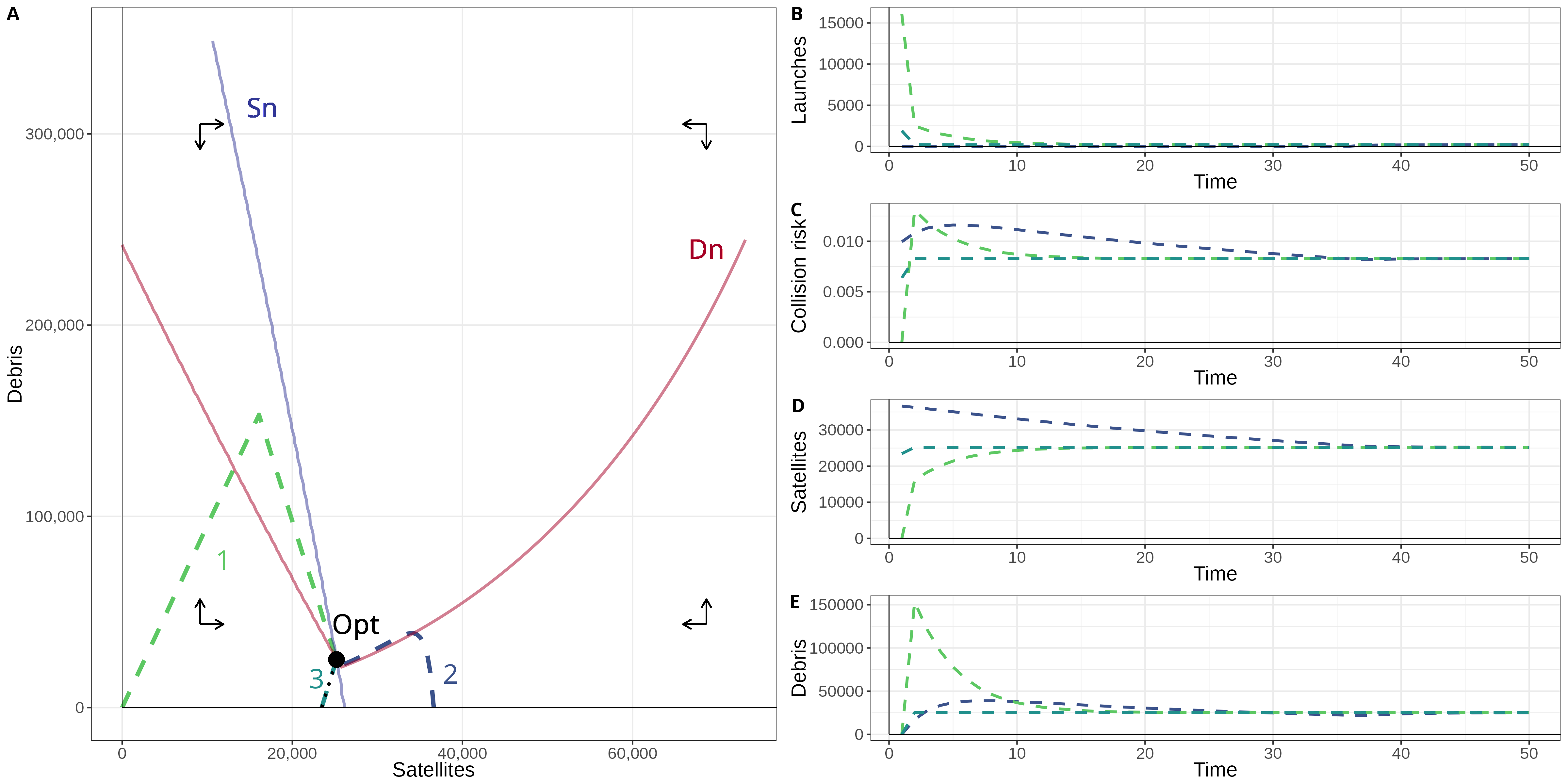} 
	\caption{\footnotesize The thick black point in panel A labeled \textit{Opt} is the stable optimal steady state. The dotted black line in panel A connecting the steady state to the y-axis is the manifold of points which can converge to the stable steady state in one step. The teal dashed line (labeled 3) along the dotted black line shows a path converging to the steady state in one step. The green dashed line (labeled 1) shows a path which overshoots in debris. The purple dashed line (labeled 2) shows a path which starts with more satellites than optimal. Panels B-E show the launch rate, collision probability, and satellite and debris stocks for sample paths 1-3 over time. Panel A also shows the satellite (blue, curve \textit{Sn}) and debris nullclines (red, curve \textit{Dn}).}
	\label{fig:overshoot-opt}
\end{figure}

\subsection{Kessler Syndrome}

We now define Kessler Syndrome in the general setting and show how it can occur under open access or the social planner. Intuitively, Kessler Syndrome is a region of the state space from which it is impossible to reach the stable steady state, given the launch policy in use.\footnote{In \citet{kessler1}, ``collisional cascading'' is defined as one of the ``possible consequences of continuing unrestrained launch activities''. ``Unrestrained launch activities'' is not defined precisely, though that paper and other engineering studies of the debris environment typically assume either continuation of historical trends indefinitely or continuation with total launch cessation at an arbitrary date \citep{rossi1998, liou2006, bradleywein2009}. Open-access launch behavior is consistent with ``unrestrained launch activities'', current legal institutions around orbit use, and economic behavior. In \cite{kessler2}, ``runaway debris growth'' is defined in terms of a launch policy which holds the stock of ``intact objects'' (active satellites and unfragmented debris objects) constant. Such a policy is less consistent with rational economic behavior---eventually the required rate of replacement should make launchers prefer to invest their funds elsewhere. In any case, our definition encompasses all such definitions and makes the dependence on the launch policy explicit.} Definition \ref{def:kessler-syndrome} formalizes this.

\begin{defn}{(Kessler Syndrome)}
	\label{def:kessler-syndrome}
	The Kessler region for a launch policy $X$ is the set of states (satellite and debris levels) such that under $X$ the debris stock will grow to infinity, i.e.
	\begin{align*}
	    \kappa_X \equiv \{ (S,D) : \lim_{\tau \to \infty} D_{\tau+1} = \infty ~|~ S_t = S, D_t = D, X_t=X(S_t,D_t), \tau \geq t\}
	\end{align*}
Kessler Syndrome occurs when the system enters the Kessler region, i.e. $(S_t,D_t) \in \kappa_X$.
\end{defn}

The ``Kessler region'' in Definition \ref{def:kessler-syndrome} is simply the complement of the stable basin. That is, Kessler Syndrome occurs when the initial condition places the system outside the stable basin defined by the launch policy and physical dynamics---when $(S,D)$ is outside the gray shaded area in figures \ref{fig:oa-policy-phase}B and \ref{fig:opt-policy-phase}B. Figures \ref{fig:oa-policy-phase} and \ref{fig:opt-policy-phase} show that the Kessler region is larger under open access than under the social planner, reflecting Proposition \ref{prop:simple-kessler}. The Kessler threshold in the simple model from Section \ref{sec:simple_econ} now takes the form of a manifold in terms of $(S,D)$ which defines the boundary between the stable basin and the rest of the state space.  

The economics of orbit use plays a key role in determining the Kessler region. Given a parameterization of the physical dynamics, the initial condition on its own is insufficient to determine whether Kessler Syndrome will occur or not. The economic institutions governing orbit use must be specified to close that loop. Open access and a social planner are two particular institutions consistent with rational forward-looking behavior. The engineering literature tends to assume different institutions, e.g. that launches will continue at a rate which maintains the current total population of active objects in orbit regardless of collision risk \citep{kessler1}, or that the launch pattern from a given interval (e.g. 1997-2004 or 2010-2017) will repeat indefinitely \citep{liou2006, lewis2020understanding}. Our model shows that Kessler Syndrome projections not grounded in a model of rational forward-looking behavior may miss important behavioral responses. For example, open-access operators will stop launching when the collision risk gets too high, potentially preventing Kessler Syndrome. On the other hand, a large-enough increase in satellite profitability may lead open-access operators to rationally choose to cause Kessler Syndrome despite prior stability. 

Figure \ref{fig:oa-kessler-excess-return} illustrates the role of economics in orbital stability through two cases: one where the discount rate is high (panels A and C), and one where the discount rate is low (panels B and D). In all panels, we report two dynamic paths, one starting with a large number of debris and no satellites (dashed arrows), and one starting with zero debris and a large number of satellites (plain arrows). In panel A, open access does not lead to runaway debris growth for the plain path, but not for the dashed path. As the discount rate drops, the open-access stable basin shrinks until it disappears and Kessler Syndrome is guaranteed in all states of the orbit under open access, including the plain one (panel B). 

For the planner things are different. Under the same parameterization of panel A, both the dashed and the plain paths are stable and converge to a steady state. As the discount rate increases, the planner places less weight on debris accumulation in the future, which means that the external cost seen in the second and third terms of equation \eqref{eqn:ss-mec} is now lower for all $S$ and $D$; panel D shows that this can lead the planner to allow Kessler Syndrome along the optimal paths in states where it was not optimal before (dashed line). At the same time, along the stable path (plain line), the planner allows a higher number of satellites and a higher number of debris compared to panel C. 

\begin{figure}[t] 
	\centering
	\includegraphics[width=0.75\textwidth]{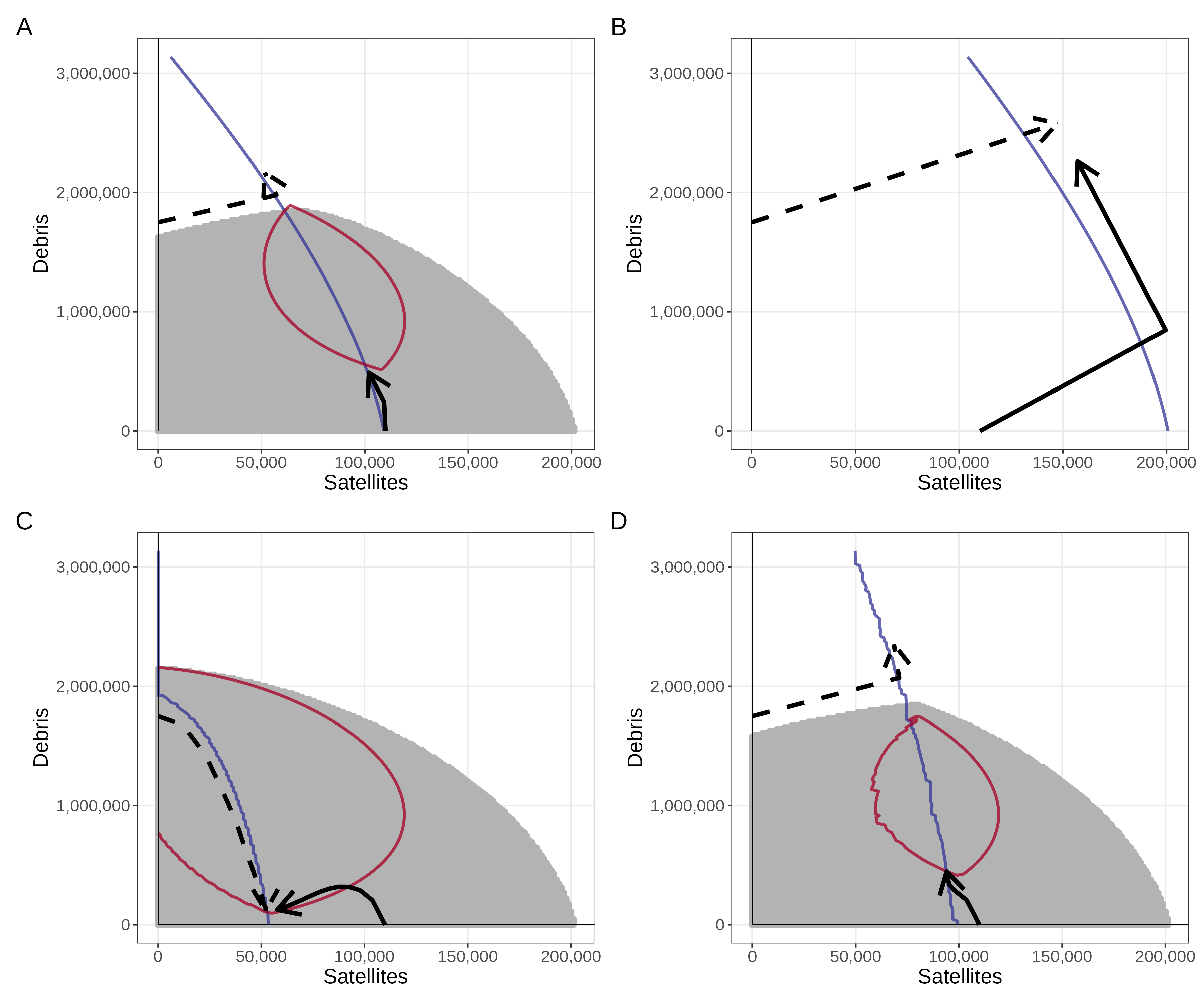} 
	\caption{\footnotesize Graphical analysis of emergence of Kessler Syndrome for paths for debris and satellites under two different inititial conditions: high debris, no satellites (dashed path), and high satellites, no debris (plain path). The upper row shows the open-access phase diagram and stable basin under two cases of launch costs and discount rates: high cost and low discount rate (low excess return) on the left, low cost and high discount rate on the right. The lower row shows the planner's phase diagram and stable basin under the same parameters. Discretization introduces small irregularities in panel D. Kessler Syndrome emerges whenever the path is outside the stable basin (gray area in the graphs).
	}
	\label{fig:oa-kessler-excess-return}
\end{figure}

The planner tends to allow Kessler Syndrome (even when they could avoid it by launching less) in states where the short-run return may justify launching, but the long-run returns will be low. Those tend to be states with very few satellites and high numbers of debris. While the planner would not reach those states had they started from a totally empty environment---$(S,D)=(0,0)$ is within the stable basin---the institutions preceding the planner's control may lead there. The history of LEO use suggests this may be plausible. From the late 1950s to the early 2000s, the majority of LEO users were government entities. This period saw the creation of the bulk of the large debris objects (i.e. rocket bodies and intact payloads) in orbit today, as the US and USSR (later Russia) conducted activities in LEO primarily for geopolitical and national security purposes. Commercial use of LEO began accelerating in the early 2000s and particularly in the late 2010s, driven by advances in miniaturization, computing, and software. Governments with non-commercial motives have continued to be large debris creators in this period---e.g. missile tests by China and Russia in 2007 and 2021---but a growing share of space traffic has been due to commercial operators. Were control of orbital space to be rationalized under a planner's regime tomorrow, the initial condition would reflect decades of geopolitical and open-access use.

\subsection{Simulation exercises}\label{sec:simulation}
Finally, to get a sense of the plausible range of times till Kessler Syndrome occurs (i.e. the system enters the runaway growth region $\kappa_X$ from Definition \ref{def:kessler-syndrome}), we calibrate the general model to the 600-650 km orbital shell and simulate open-access behavior forward for 200 years from the 2020 initial condition, $(S_{2020},D_{2020}) = (158, 626)$. We assume there are no changes in institutions, no debris removal technologies deployed, and no further anti-satellite missile tests.\footnote{There are several technological and institutional challenges associated with implementing debris removal, including current lack of feasibility and free-rider issues \citep{weedenIAC, klima2016space, klima2018space}. Anti-satellite missile tests have been a major source of debris so far \citet{lioujohnson2009, oltrogge2022iaa, pardini2023short}. Although the US has recently announced a moratorium on their own use of such tests, there is only limited progress towards a binding international agreement preventing their use \citep{panda2022us, foust2022united}.} To account for time-varying payoffs and potential market competition or economies of scale in operating satellites, we rewrite the per-period payoff as $p_t(S_t)$, where
\begin{equation}
    p_t(S_t) = \pi e^{at} (1 + \eta) S_t^{\eta},
    \label{eqn:time-varying-payoff}
\end{equation}
In the model so far we have taken $a = 0$ and $\eta = 0$, so $p_t(S_t) = \pi$. Equation \ref{eqn:time-varying-payoff} generalizes this payoff. The payoff consists of two components: factor productivity $\pi e^{at}$, and output elasticity $(1 + \eta) S^{\eta}$. The factor productivity parameter $\pi e^{at}$ represents the productivity of a satellite in an orbital slot at a reference time with exponential growth since then. The elasticity parameter $\eta$ reflects the occupancy elasticity of satellite output. 
We also allow time-variation in the per-period cost $F$, which we model as $F_t = c e^{bt}$. We calibrate the parameters $p, c$ such that when $\eta = 0$, $p_t(S_t)$ and $F_t$ match the 2020 aggregate revenues and costs of the satellite sector as reported by \citet{tsr2021}, and the parameter $b$ to match the empirical growth rate of those sectoral costs (2.5\%). We also modify the laws of motion to allow for imperfect collision avoidance maneuvers and limited satellite lifetimes. We describe the data and calibration procedures, as well as the modified laws of motion, in more detail in Appendix \ref{appendix:calibration}. \\

\noindent \textsc{Role of Autocatalytic Growth.} The first simulation exercise evaluates the claim that it is the component of autocatalytic growth in the debris growth function $G(S,D)$, represented by the partial derivative $G_D(S,D)$, that primarily matters for Kessler Syndrome, as implied by Proposition \ref{multipleOASS}. We simulate the model with values of the parameter controlling the new fragments from the debris-debris collision, $\beta_{DD}$ in equation \eqref{eqn:general-new-debris-function-example}, that range from 0 to 500, and for each, we compute the time at which the system enters the Kessler region, which we referred to as the ``Kessler time.'' Recall that entering the Kessler region does not imply that the orbit is unusable, but rather that the orbit has crossed a threshold beyond which eventual unusability is guaranteed.

Figure \ref{fig:oa-kessler-times}A reports these simulation results. Focusing on the black line, for any value above 450 fragments, the system is already in the Kessler region in 2023, the time of this writing; while for values below 120-150, the Kessler time exceeds 3000. In between, the Kessler time ranges between 2023 and 2130 for $\beta_{DD}$ between 450 and 400 fragments, with a steep slope that implies a high sensitivity of the stability of the system in that parameter range.  For values of $\beta_{DD}$ below 400, the decrease in Kessler time is more gradual. The middle (dark gray) line in the figure is computed with a value of 332 for the new fragments from satellite-debris collisions, the parameter $\beta_{SD}$ in equation \eqref{eqn:general-new-debris-function-example}. The two additional curves in the figure are computed with a higher and lower value of the same parameter, 100 and 600. The three curves essentially sit on each other for values of $\beta_{DD}$ above 200. In particular, they are indistinguishable for values between 400 and 450. The implication is that the Kessler time is very insensitive to the satellite-debris collision part of the debris growth function compared to the autocatalytic part represented by $\beta_{DD}$. The numerical results of figure \ref{fig:oa-kessler-times}A thus confirm the theoretical result in Proposition \ref{multipleOASS}: in the presence of endogenous launching decisions, Kessler Syndrome is primarily driven by the autocatalytic part of debris growth.\\

\noindent \textsc{Role of Returns Growth and Output Elasticity.} The second simulation exercise evaluates the impact of the growth rate of returns ($a$) and the output elasticity ($\eta$) on the emergence of Kessler Syndrome. There is as yet insufficient data on the satellite sector to conduct a detailed demand estimation exercise and identify the components of the occupancy elasticity, and the future growth rate of payoffs from satellite operation is uncertain. We therefore conduct a simple sensitivity analysis over a small range of occupancy elasticities, namely $\eta \in\{-0.2,-0.1,0\}$, consistent with the range identified in \citet{raoletizia2021},\footnote{\citet{raoletizia2021} estimate a discrete-choice model of sorting over orbits with satellite and debris stocks as orbital characteristics while holding the profile of launches fixed, finding both positive and negative responses to satellites in the same shell for different types of operators---commercial operators seem to prefer to avoid shells with other satellites, while civil government and military operators seem to prefer those shells. However, these substitution and complementarity effects cannot be separately identified in the aggregate $\eta$ we model without more granular data on revenues and costs for individual LEO satellite operators or a credible identification strategy. To the best of our knowledge the former is not available, and the latter will require further research and is beyond our scope here.} and consider a range of payoff growth rates $a$ between $2.5\%$ and $10\%$. 

Figure \ref{fig:oa-kessler-times}B shows the estimated Kessler times over the annual returns growth rate. Focusing on the black curve, which corresponds to the benchmark case of $\eta=0$, the results show that Kessler time is highly sensitive to the growth rate of payoffs. For rates above 8\%, the Kessler time is between 2040-2045, while for rates below 3\% it is larger than 2184. The impact of a negative elasticity is represented by the two additional curves, which show that Kessler time is increased across all growth rates, but in a smaller proportion for high growth rates. Specifically, for rates below 3\% the Kessler time is 2389 for $\eta=-0.1$, and 2592 for $\eta=-0.2$, while for rates above 8\% it is between 2060-2073 and 2079-2097, respectively. Overall, under the current parameterization, a negative elasticity has a positive impact on the Kessler time because it functions as a dampener on the incentive to launch when the stock of satellites increases. However, we emphasize that this result is not general. As we show in Appendix \ref{appendix:downward-demand}, a moderately negative elasticity might have a positive effect on the incentive to launch, with an ambiguous impact on the Kessler time. We leave the exploration of the interaction between output elasticity and Kessler time to future work.

\begin{figure}[t] 
	\centering
	\includegraphics[width=0.93\textwidth]{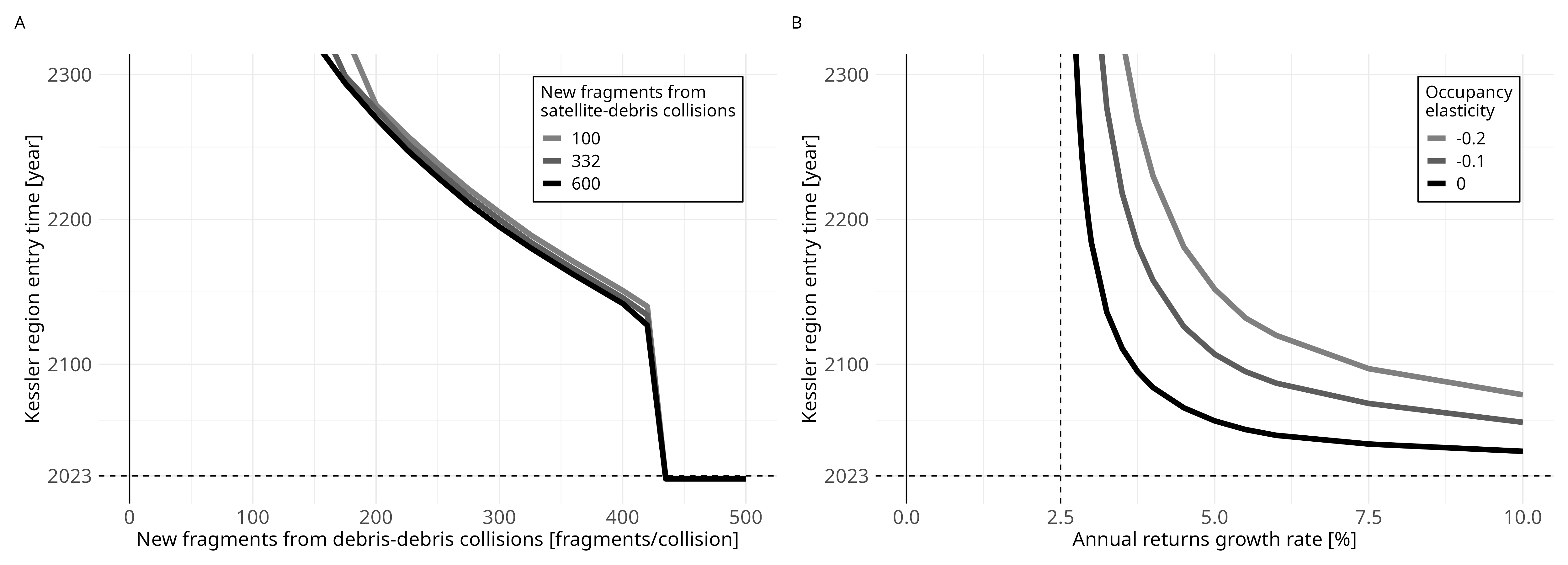} 
	\caption{\footnotesize Simulated Kessler Syndrome region entry times in the 600-650 km shell under different numerical specifications. Panel A plots Kessler times against debris-debris fragmentation, $\beta_{DD}$ in equation \eqref{eqn:general-new-debris-function-example}, for three values of satellite-debris fragmentation, $\beta_{SD}\in\{100, 332, 600\}$, and with $a=3\%$ and $\eta=0$. Panel B plots Kessler times against unitary satellite returns growth rates for three values of the output elasticity, $\eta \in\{-0.2,-0.1,0\}$, and with $\beta_{DD} = 326$ and $\beta_{SD} = 332$. In both panels the growth rate of satellite costs is fixed at 2.5\%. Details of the calibration are reported in Appendix \ref{appendix:calibration}.
	}
	\label{fig:oa-kessler-times}
\end{figure}

\section{Discussion and conclusion}
\label{sec:discussion}

\subsection{Discussion}
Orbit use bears some similarities to terrestrial resources, though there are also some significant differences. In this section we briefly discuss some of these similarities and differences. Where possible we refer to the simple model.

\paragraph{The source of orbit-use externalities:}

The fundamental coordination problem in orbit use is due to the lack of property rights over orbital slots and the potential for collisions between satellites. The OST thus induces open access to orbit. This open access creates orbit-use externalities due to collision risk, like congestion in a high-seas fishery. Though debris is a source of pollution, its damages are also due to the risk of collisions---not harms to consumers. All orbital externalities are thus reciprocal, and to some degree limited by the same incentives that drive operators to use the resource in the first place. But just as fishers in a high-seas fishery cannot claim exclusive rights over specific fish in a finite population, individual operators in low-Earth orbit cannot claim exclusive rights over specific ``slots'' in a finite (if large) set. Collisional rent dissipation induces resource overuse even as it is limited by expected economic returns \citep{gordon1954}. 

Equation \eqref{eqn:PLN} shows the centrality of collision risk in the externality problem: if the collision risk is exogenous, i.e. $q(\cdot) \to q$, then $q'(\cdot) \equiv 0$ and the planner's first-order condition matches the equilibrium condition. The externality would disappear as social and private marginal costs align. Slots would be claimed until marginal revenues equaled marginal costs and the allocation, whatever risks it entailed, would be economically efficient. But when collision risk is endogenous and marginal satellites do not internalize their effects on the fleet, the open-access equilibrium is inefficient. Note that open-access firms \emph{do} internalize collision risk directly, but without considering how their actions affect the risks facing others.\footnote{I.e. $q(\cdot)$ but not $q'(\cdot)$ is present in equation \eqref{eqn:EXP}.} Again, there is a parallel to open-access fishers accounting for the current stock level due to the effort it implies, but not how their harvests change future population growth. While the importance of collision risk can be seen in prior economic literature on orbit use (e.g. \citet{aac2015, rouillon2019, raoetal2020}), the simple model shows clearly how its endogeneity drives the externality even without debris ($\sigma=0$).

\paragraph{Hotelling's Rule:} Hotelling's Rule for optimal resource extraction requires that the rate of growth in the profits from resource extraction match the rate of interest. Here, Hotelling's Rule manifests in a requirement that the expected rate of return on capital invested in orbit use equals the discount rate (equation \ref{eqn:zeroprofit-rate}). There are two main differences between orbit use and resource extraction: 
\begin{enumerate}
    \item Increased orbit use reduces the survival rate of satellites and thus their expected revenues. This effect exists even if it does not decrease the price received by slot users. The problem is also not necessarily finite. In the absence of autocatalytic debris growth, orbital slots cannot be permanently exhausted---filled slots are ``replenished'' at the net rate of satellite demise and orbital decay.
    \item The endogenous risk of collision reduces the rate of return on marginal fleet assets (i.e. the rate of return on the fleet from adding a marginal satellite) below the rate of return on marginal satellite assets (i.e. the rate of return on the marginal satellite added). Thus, while both open access firms and the planner satisfy a version of Hotelling's Rule, the open-access condition (equation \ref{eqn:OPA} in the simple model, equation \ref{eqn:zeroprofit} in the full model) considers only a marginal satellite whereas the planner's optimality condition (equation \ref{eqn:PLN} in the simple model, equation \ref{eqn:optflowcond} in the full model) considers the marginal return on the fleet as a whole. A rule similar to the open-access condition here appears in the problem of drilling subject to pressure constraints studied in \cite{anderson2018hotelling}, where the decreasing rate of well pressure plays a similar role as the decreasing survival rate here.
\end{enumerate}

\paragraph{Discounting and renewable resources:} The fact that orbits can be renewed through satellite demise and orbital decay creates similarities to renewable terrestrial resources like fisheries. As in terrestrial renewable resources, the extent of orbit use under open access exceeds the extent of fleet-wide profit-maximizing use due to the lack of property rights over the underlying natural resource (fish in fisheries, slots in orbits). However, the role of the discount rate differs substantially. Biological assets like fish can reproduce without human investment, allowing a user to profit today at the expense of tomorrow by investing and harvesting many members of the species. Artificial assets like satellites require upfront investment and only generate payoffs over time---investing more today may reduce tomorrow's payoffs due to collisions and won't increase the payoffs received today. Changes in the discount rate thus have the opposite effect on orbit use compared to biological renewable resources. An increase in the discount rate can increase the rate of harvest in a fishery by diminishing the perceived future costs of reduced harvest levels (e.g. see equation (2.11) in \cite{clarkmunro1975}), but will decrease the rate of satellite launches by diminishing the perceived future benefits satellites deliver (i.e. increasing $r$ in equations \eqref{eqn:OPA} or \eqref{eqn:PLN} will decrease the number of launches).

\paragraph{Profit-maximizing extinction:} Finally, Kessler Syndrome has similarities to extinction of a biological population. When a biological population drops below its critical depensation level, it will eventually go extinct whether it is harvested or not. Similarly, once the debris population exceeds a critical threshold, orbital slots will become increasingly full of debris objects and unusable whether new satellites are launched or not. We observe similarities to results in \cite{clark1973profit} regarding optimal extinction of biological populations. As in \cite{clark1973profit}, even the planner can find it optimal to cause Kessler Syndrome if the returns from harvesting the resource over a finite interval exceed the present value of extending the length of the interval indefinitely. However, open-access firms will be ``more likely'' to cause Kessler Syndrome than the planner would. Intuitively, open access means that firms who might receive long-run returns from sustainable orbit use are unable to exclude marginal entrants from launching satellites and crossing the Kessler threshold. Since \emph{ceteris paribus} larger discount rates reduce the intensity of orbit use they also reduce the chances of Kessler Syndrome---again, a departure from the case of biological populations. 

\subsection{Conclusion}
\label{sec:conclusion}

In this paper we present a dynamic physico-economic model of orbit use under rational expectations with endogenous collision probability and Kessler Syndrome. We show how both economics and physics drive equilibrium and optimal orbital-use patterns, derive the external cost of a marginal satellite, and examine conditions under which Kessler Syndrome can be an equilibrium or optimal outcome. We highlight three messages regarding orbital-use management.

First, under open access too many firms will launch satellites because they won't internalize the risks they impose on other orbit users. This inefficiency is independent of whether Kessler Syndrome is possible or not, or even whether debris exists or not. Second, Kessler Syndrome is possible when there is autocatalytic debris growth. All else equal, Kessler Syndrome is more likely to be an equilibrium outcome as the rate of return on a satellite rises, even if firms respond to orbital congestion by launching fewer satellites. Third, the role of the discount rate in orbit use is unique. Unlike typical bioeconomic commons problems, \emph{ceteris paribus} higher discount rates induce less (rather than more) open-access and optimal resource use. But if the cost of deploying a satellite is low and the discount rate is high, Kessler Syndrome may maximize the net present value of orbit use over the long run, even when all external costs are internalized.

Commons management is one of the oldest problems in economics. Economists tend to favor property rights, corrective taxes, or other market-based mechanisms to solve them. While these mechanisms can ensure efficient orbit use, more data collection and research is needed to understand how orbital-use management policies should be designed in light of its unique features. Whether they are enforced by states, structured as self-enforcing agreements between private actors, or some combination of the two, effective orbital-use management policies must address the open access problem.

\newpage
{
	\setlength{\bibsep}{3pt}
	\setstretch{1}
	\bibliography{database}
	\bibliographystyle{jpe}
}

\begin{appendices}

\section{Downward-sloping demand for satellites}
\label{appendix:downward-demand}
In this Appendix, we extend the simple model to the case of per-period returns that depends negatively on the stock of satellites in orbit. Our goal is to show that Kessler Syndrome is still possible, with the conditions for its emergence appropriately modified. We begin by considering a technology that uses satellites to produce output. This output is an aggregate bundle of goods and services provided by different types of satellites, e.g. a composite good incorporating telecommunications, imaging, etc. We normalize the price of the composite output to $1$, and the unitary cost of a satellite input as $\pi>0$. The representative aggregator firm takes the unitary cost of a satellite as given and maximizes the per-period profits
\begin{equation*}
    \pi Z_t^{1+\eta} - p_t Z_t.
\end{equation*}
where $Z_t$ denote the number of satellites \textit{operating} at time $t$. The solution to the maximization problem corresponds to the demand faced by the satellite operators:
\begin{equation*}
    \pi(1+\eta)Z_t^{\eta} = p_t.
\end{equation*}
The case considered in the main text is equivalent to a situation where the aggregation technology is constant returns to scale, so that $\eta=0$, and $p_t = \pi$. The unitary return of a satellite can still change over time, something that we do assume in the fully dynamic model, but the change is exogenous to the stock of satellites, and it corresponds to a time-varying productivity $\pi$.

Let us consider the case of $\eta<0$, which implies that the more satellites operating in orbit, the lower the unitary return. An immediate implication of this negative relationship is that as the orbit fills with satellites, the return on satellites declines, and so does the incentive to launch additional satellites. It may appear that a downward-sloping demand curve makes it more difficult to congest the orbit and obtain Kessler Syndrome. However, this argument is incomplete. It is important to recognize that the orbit congestion depends on the total number of objects in orbit, while the unitary return depends on the objects in orbit that are \textit{still} operating. In other words, as the number of objects in orbit increases, the unitary return might also be increasing if the increase in the number of orbiting objects is primarily due to the increase in debris while operating satellites decline!

To see this in the context of our simple model, consider that in period $t=s$ the number of operating satellites is
\begin{equation}
Z = q(S)S
\end{equation}
while the total number of objects in orbit---which matters for the survival probability---is the number of satellites launched at $t=0$, $S$. So the unitary return in period $t=s$ decreases as $S$ increases only if
\begin{equation}
\frac{dZ}{dS} = q'(S)S + q(S) \geq 0. \label{downward-condition}
\end{equation}
If this condition does not hold, as more satellites are launched, the impact on the survival probability dominates the impact on the operating satellites, and the per-period returns are actually increasing in $S$, even when there is a downward-sloping demand for satellites. If we let $q(X) = 1-S/\bar{X}$, the condition above is violated whenever $S>\frac{1}{2}\bar{X}$, that is whenever the orbit is sufficiently congested. The insight offered by this simple example is that a downward-sloping demand for satellites, when combined with collision risk, might end up exacerbating the incentive to congest the orbit. We believe this positive feedback mechanism is interesting, but we leave a systematic analysis of its implications for future work. For the purpose of the current analysis we consider the argument above as a reassurance that a constant return $\pi$ is not an unreasonable assumption for our baseline model.

Assuming that operators face the demand curve above, one can show that the condition for Kessler Syndrome under open access becomes
\begin{equation}
    F\leq \frac{\pi(1+\eta)}{1+r}S_K^{\eta}q(S_K)^{1+\eta}.
\end{equation}
Compared to the case of $\eta=0$, Kessler Syndrome is more likely to occur in the presence of a downward-sloping demand when
\begin{equation}
    \left[S_K q(S_K)\right]^{\eta}>\frac{1}{1+\eta}.
\end{equation}
Using once again the functional form $q(X) = 1-S/\bar{X}$, this condition corresponds to
\begin{equation}
    \left[\bar{X}\left(\frac{\sqrt{1+4\sigma}-1}{2\sigma}\right)\left(\frac{1+2\sigma - \sqrt{1+4\sigma}}{2\sigma}\right)\right]^{\eta}>\frac{1}{1+\eta}.
\end{equation}
Numerical computations show that the impact of $\eta$ is non-monotonic. Setting $\bar{X}=1$ and $\sigma=1$ the inequality above holds approximately for $\eta\in(-0.5,0)$, with a peak in the gap at around $-0.3$, which means that for a moderately downward-sloping demand, the positive feedback mechanism described above is strongest. 

Taken together, the results just presented indicate that the introduction of a downward-sloping demand has an ambiguous effect on the emergence of Kessler Syndrome. Our maintained assumption of a constant return $\pi$ corresponds to balancing the two contrasting effects highlighted above. The positive effect articulated here is an interesting and potentially important extension of our analysis that we leave to future work.

\section{Proofs and derivations}

\subsection{Proofs omitted from main text}

\simpleKessler*

\begin{proof}
First, recall the definition of Kessler Syndrome in the simple model: a launch rate $S$ such that $S<\bar{X}$ while $g(S) \geq \bar{X}$. The smallest level of $S$ at which this condition can hold is $S_K$, since $g(S_K) = \bar{X}$, $g(S)$ is increasing in $S$ when $\sigma > 0$, and by assumption $S_K < \bar{X}$. \\

Suppose open access will cause Kessler Syndrome. Then $\hat{S}$ must be such that $\hat{S}<\bar{X}$ while $g(\hat{S}) \geq \bar{X}$, implying $\hat{S} \geq S_K$. The equilibrium condition becomes
\begin{equation}
    F = q (\hat{S}) \frac{\pi}{1 + r}.
\end{equation}

Since $q$ is decreasing in $S$, the above condition can be satisfied if and only if
\begin{equation}
    F \leq q (S_K) \frac{\pi}{1 + r}.
\end{equation}

Next, suppose the social planner will cause Kessler Syndrome. Then $S^*$ must be such that $S^*<\bar{X}$ while $g(S^*) \geq \bar{X}$, implying $S^* \geq S_K$. The optimality condition becomes
\begin{equation}
    F = [q (S^*) + S^* q' (S^*) ] \frac{\pi}{1 + r}.
\end{equation}

Since $q$ is decreasing in $S$, the above condition can be satisfied if and only if
\begin{equation}
    F \leq [q (S_K) + S^* q' (S_K) ]  \frac{\pi}{1 + r}.
\end{equation}

This completes the proof.
\end{proof}

\multiplicityOASS*

\begin{proof}
    The proposition asserts:
    \begin{enumerate}
        \item \emph{Existence of multiple steady states:} Given a positive excess return on a satellite, multiple open-access steady states can exist if debris objects can collide and produce new debris.
        \item \emph{Stability of steady states:} An open-access steady state will be stable if and only if
        \begin{equation}
            (G_D(S^*,D^*) - \delta) < \frac{L_D(S^*,D^*)}{L_S(S^*,D^*)}(G_S(S^*,D^*) + m(\frac{\pi}{F} - r)). \tag{\eqref{eqn:stability-condition}}
        \end{equation}
        \item \emph{Ordering of steady states:} When $G$ is strictly convex in both arguments and two steady states exist, the higher-debris is unstable.
    \end{enumerate}

Before proving them, we establish a useful reduction.

\paragraph{0. A useful reduction:} The open-access steady states are defined by equations \eqref{eqn:satLoM}, \eqref{eqn:debLoM}, and \eqref{eqn:zeroprofit}, combined with the conditions $D_t = D_{t+1} = D$ and $S_t = S_{t+1} = S$. Since $L$ is monotone increasing in both arguments it is invertible, and equation \eqref{eqn:zeroprofit} implicitly determines the number of satellites as a function of the amount of debris, the excess return on a satellite, and the collision rate function,
	\begin{equation}
	L(S,D) = \frac{\pi}{F} - r \implies S = S(\frac{\pi}{F} - r, D). \label{eqn:SSeqmcond}
	\end{equation}
	Since $L$ is monotone increasing in each argument, $S(\frac{\pi}{F} - r, D)$ is monotone decreasing in $D$. Since $S$ must be nonnegative, there exists a nonnegative $D^S : S(\frac{\pi}{F} - r, D) = 0 ~ \forall D \geq D^S$. Let $\hat{S}$ be the equilibrium satellite stock as a function of the debris stock. So we have
	\begin{align}
	    \hat{S} =
	    \begin{cases}
	        S(\frac{\pi}{F} - r, D) \text{ if } D \in [0,D^S) \\
	        0 \text{ if } D \geq D^S
	    \end{cases}
	\end{align}
	
	Using $\hat{S}$ we can reduce equations \eqref{eqn:satLoM}, \eqref{eqn:debLoM}, and \eqref{eqn:zeroprofit} to a single equation in debris,
	\begin{equation*}
	\mathcal{Y}(D) = -\delta D + G(\hat{S},D) + m (\frac{\pi}{F} - r) \hat{S}, \label{eqn:OASSredux}
	\end{equation*}
	
	with the solutions
    \begin{equation}
    \{\hat{D} \geq 0 : \delta \hat{D} = G(\hat{S},\hat{D}) + m (\frac{\pi}{F} - r) \hat{S}\}
	\end{equation}
	being the open-access steady states.

\paragraph{1. Existence of multiple steady states:} Using the above reduction, we focus our attention on solutions to equation \eqref{eqn:OASSredux}.  $\delta D$ is monotonically increasing in $D$ with $\delta D = 0$ when $D = 0$, and $m(\frac{\pi}{F}-r)\hat{S}$ is monotonically decreasing in $D$ with $\hat{S}>0$ when $D=0$, but $\hat{G} \equiv G(\hat{S},D)$ is nonmonotone in $D$. To see this, note
	\begin{align}
	    \frac{d\hat{G}}{dD}(\hat{S},D) &= \underbrace{\frac{\partial G}{\partial S}}_{\geq0} \underbrace{\frac{\partial \hat{S}}{\partial D}}_{\leq 0} + \underbrace{\frac{\partial G}{\partial D}}_{\geq 0}, \text{ with}\\
	   \frac{d\hat{G}}{dD}(\hat{S},0) &= \frac{\partial G}{\partial S} \frac{\partial \hat{S}}{\partial D} < 0 \text{ and} \\
	   \frac{d\hat{G}}{dD}(0,D^S) &= \frac{\partial G}{\partial D} > 0 ,
	\end{align}

where $\frac{\partial \hat{S}}{\partial D} = - \frac{L_D}{L_S} \leq 0$ by application of the Implicit Function Theorem on equation \eqref{eqn:zeroprofit}.

Let $\hat{D}$ be a solution to equation \eqref{eqn:OASSredux}. If $G_D > 0$, then $\hat{G}$ is nonmonotone in $D$ and the existence or uniqueness of $\hat{D}$ cannot be guaranteed. If $G_D$ is large enough, $\hat{D}$ will not exist; if $G_D$ is not too small, multiple $\hat{D}$ will exist. If $G_D = 0$, then the existence of $\hat{D}$ also ensures its uniqueness. If $G_D$ is strictly convex in both arguments, at most two $\hat{D}$ can exist.

\paragraph{2. Stability of steady states:} Since $\mathcal{Y}(D)$ is a reduction of the open-access dynamical system, its fixed points are isomorphic to the fixed points of equations \eqref{eqn:satLoM}, \eqref{eqn:debLoM}, and \eqref{eqn:zeroprofit}. The sign of $\frac{\partial \mathcal{Y}}{\partial D}$ at solutions to $\mathcal{Y}(D) = 0$ matches the sign of the respective eigenvalues of the full system.

Applying the Implicit Function Theorem to equation \eqref{eqn:SSeqmcond} to calculate $S_D$ and then differentiating $\mathcal{Y}$ in the neighborhood of an arbitrary solution $D^*$, we obtain
	\begin{equation}
	\frac{\partial \mathcal{Y}}{\partial D}(D^*) = (G_D(S^*,D^*) -\delta) - \frac{L_D(S^*,D^*)}{L_S(S^*,D^*)}(G_S(S^*,D^*) + m(\frac{\pi}{F} - r)),
	\end{equation}
where $S^* \equiv S(\frac{\pi}{F} - r, D^*)$. Both $G_S(S^*,D^*)$ and $m(\frac{\pi}{F} - r)$ are positive by assumption. So $\frac{\partial \mathcal{Y}}{\partial D}(D^*) < 0$ holds if and only if $\delta$ is small enough, or $\frac{L_D(S^*,D^*)}{L_S(S^*,D^*)}$ is large enough, i.e.

\paragraph{3. Ordering of steady states:} When $G$ is strictly convex in both arguments and $\mathcal{Y}(D) = 0$ has two solutions. Denote the smaller solution by $\ubar{D}$, and the larger solution by $\bar{D}$. The curve $G(\hat{S},D) + m (\frac{\pi}{F} - r) \hat{S}$ is above $\delta D$ when $D = 0$, and again as $D \to \infty$. $\mathcal{Y}(D)$ must therefore approach 0 from above as $D \to \ubar{D}$ from the left, and from below as $D \to \bar{D}$ from the left. This implies that at $\ubar{D}$,
	\begin{equation}
	\frac{\partial \mathcal{Y}}{\partial D}(\ubar{D}) = (G_D(\ubar{S},\ubar{D}) -\delta) - \frac{L_D(\ubar{S},\ubar{D})}{L_S(\ubar{S},\ubar{D})}(G_S(\ubar{S},\ubar{D}) + m(\frac{\pi}{F} - r)) < 0
	\end{equation}
	
and at the second solution, $\bar{D}$,
	\begin{equation}
	\frac{\partial \mathcal{Y}}{\partial D}(\bar{D}) = (G_D(\bar{S},\bar{D}) -\delta) - \frac{L_D(\bar{S},\bar{D})}{L_S(\bar{S},\bar{D})}(G_S(\bar{S},\bar{D}) + m(\frac{\pi}{F} - r)) > 0.
	\end{equation}
where $\ubar{S} = \hat{S}(\ubar{D})$ and $\bar{S} = \hat{S}(\bar{D})$.
\end{proof}

\overshootingOASS*

\begin{proof}
    We first define the following sets and functions, where $S,D \geq 0$ is assumed:
\begin{itemize}
    \item The action region: the set of states with positive open-access launch rates,
    \begin{equation}
        A \equiv \left\{ (S,D) : \frac{\pi}{F} - r - L(S',D') \geq 0 \right\},
    \end{equation}
    where 
    \begin{align*}
        S' &= S(1-L(S,D)) + X \\
        D' &= D(1-\delta) + G(S,D) + mX,
    \end{align*}
    
    \item The equilibrium manifold:
    \begin{equation}
        E \equiv \left\{ (S,D) : \frac{\pi}{F} - r - L(S,D) = 0 \right\}.
    \end{equation}
    \item The stable open-access steady state: following the reduction used in the proof of Proposition \ref{multipleOASS}, we characterize the stable open-access steady state as 
    \begin{align}
        E_s \equiv \bigg\{(\hat{S},D) &: \mathcal{Y}(D) = - \delta D + G(\hat{S},D) + m(\frac{\pi}{F} - r)\hat{S}=0, \nonumber \\
        \hat{S} &: L(\hat{S},D) = \frac{\pi}{F} - r,~~ \mathcal{Y}'(D) < 0 \bigg\}.
    \end{align}
     \item The physical dynamics: the mapping $P_{SD} : \mathbb{R}^2_+ \to \mathbb{R}^2_+$ which describes the effect of orbital mechanics on the satellite and debris stocks in one period,
    \begin{align}
        P_{SD}(S,D) & \equiv \left( S(1-L(S,D)) , ~~ D(1-\delta) + G(S,D) \right).
    \end{align}
    
    \item The one-step set: the set of states from which one period's physical dynamics, followed by launching, will read an open-access steady state,
     \begin{equation}
        A_{P1} \equiv \{ (S,D) : P_{SD}(S,D)+(X, mX) \in E_s,~~ X \in (0,\bar{X}] \}.
    \end{equation}
    
    \item The one-step ray: the set of states from which one period of launching will reach an open-access steady state,
    \begin{equation}
        A_1 \equiv \left\{ (S,D) : (S+X, D+mX) \in E_s,~~ X \in (0,\bar{X}] \right\},
    \end{equation}
    where $m$ is the same as in the debris law of motion. The one-step ray can be viewed as part of a decomposition of the satellite and debris laws of motion: after a period's physical dynamics have been applied, launches to the stable steady state occur from the one-step ray. The one-step set encompasses both of these components.
\end{itemize}

Our proof proceeds in three steps. First, we show that initial conditions in the action region $A$ reaching points on the equilibrium manifold $E \setminus E_{S}$ must overshoot an open-access steady state. Second, we establish the bijectivity of the physical dynamics $P_{SD}$. Third, we show that these results imply that the one-step set $A_{P1}$ has zero Lebesgue measure on $A$. \\

\textbf{1. Initial conditions in the action region $A$ reaching points on the equilibrium manifold $E \setminus E_{S}$ must overshoot an open-access steady state:} Since the launch rate constraint does not bind, any point in $A$ will by definition reach a point in $E$. Since $E_s$ contains at most one element given the strict convexity of $G$ while $E$ is a manifold, $E_s \subset E$. Given that $L$ is increasing in both arguments, points in $E \setminus E_s $ must therefore have either larger $S$ and smaller $D$ than $E_s$, or vice versa. Consequently, reaching points in $E \setminus E_s$ constitutes overshooting $E_s$ in one state variable and undershooting in the other. \\

\textbf{2. Bijectivity of the physical dynamics $P_{SD}$:} To show that $P_{SD}$ is a bijection on $\mathbb{R}^2_{+}$, we separate the physical dynamics into two functions $P_S, P_D : \mathbb{R}^2_+ \to \mathbb{R}_+$,
\begin{align}
    P_S(S,D) &= S(1-L(S,D)), \\
    P_D(S,D) &= D(1-\delta) + G(S,D)  .
\end{align}

$P_D$ is a sum of strictly monotone increasing functions, so is strictly monotone increasing as well. Strictly monotone functions are bijections, so $P_D$ is a bijection. So for two arbitrary pairs $(S_1,D_1)$ and $(S_2,D_2)$ we have
\begin{equation}
    P_{SD}(S_1,D_1) = P_{SD}(S_2,D_2) \iff P_{S}(S_1,D_1) = P_{S}(S_2,D_2) ~\&~ P_{D}(S_1,D_1) = P_{D}(S_2,D_2)
\end{equation}
$P_S$ is a function, so we have $(S_1,D_1) = (S_2,D_2) \implies P_{S}(S_1,D_1) = P_{S}(S_2,D_2)$, but since $S L(S,D)$ may be non-monotone the other direction may not hold. Since $P_D$ is a bijection, $P_{D}(S_1,D_1) = P_{D}(S_2,D_2) \iff (S_1,D_1) = (S_2,D_2)$. Putting this together we have the following:
\begin{itemize}
    \item If $(S_1,D_1) = (S_2,D_2)$, then $P_{SD}(S_1,D_1) = P_{SD}(S_2,D_2)$.
    \item If $P_{SD}(S_1,D_1) = P_{SD}(S_2,D_2)$, then $P_{S}(S_1,D_1) = P_{S}(S_2,D_2)$ and $P_{D}(S_1,D_1) = P_{D}(S_2,D_2)$. While there may exist a pair $(S_1,D_1) \neq (S_2,D_2)$ such that $P_{S}(S_1,D_1) = P_{S}(S_2,D_2)$, the bijectivity of $P_D$ means $P_{D}(S_1,D_1) \neq P_{D}(S_2,D_2)$. 
\end{itemize}

Consequently, $P_{SD}(S_1,D_1) = P_{SD}(S_2,D_2)$ if and only if  $(S_1,D_1) = (S_2,D_2)$, i.e. $P_SD$ is a bijection on $\mathbb{R}^2_{+}$. \\

\textbf{3. The one-step set $A_{P1}$ has zero Lebesgue measure on $A$:} 

By definition, $A_1 \subseteq A$. Since $E_s$ contains at most one element, $A_1$ is a single line segment, so $A_1 \subset A$. The Lebesgue measure on $A$ of $A_1$ is therefore zero.

Since $P_{SD}$ is a bijection, the Lebesgue measure of the pre-image of $A_1$ under $P_{SD}$,
\[P^{-1}_{SD}(A_1) \equiv \left \{ (S,D) : P_{SD}(S,D) \in A_1 \right \},\]
is the same as the Lebesgue measure of $A_1$. Since $P^{-1}_{SD}(A_1) = A_{P1}$, the Lebesgue measure on $A$ of $A_{P1}$ is also zero. Lebesgue measure is isomorphic to any non-atomic probability measure, so the one-step set is measure zero under any non-atomic probability measure. This gives the desired result: initial conditions with positive open-access launch rates will overshoot the stable open-access steady state except on a set of measure zero.

\end{proof}

\subsection{Optimal launch policy and external cost}
\label{appendix:MECderivation}
The infinite-horizon sequence version of the fleet planner's problem is
\begin{align}
    \max_{\{X_t, S_{t+1}, D_{t+1}\}_{t=0}^{\infty}}& S_t Q(S_t,D_t,X_t) + \frac{1}{1+r} \sum_{\tau=t}^{\infty} \frac{1}{1+r}^{\tau - t - 1}X_{\tau}(\frac{1}{1+r} Q(S_{\tau+1},D_{\tau+1},X_{\tau+1}) - F) \label{eqn:plan-sequence-problem}\\
\text{s.t. }    
    Q(S_t,D_t,X_t) &= \pi + \frac{1}{1+r}(1 - L(S_t,D_t)) Q(S_{t+1}, D_{t+1}, X_{t+1}) \\
    S_{t+1} &\leq S_t(1-L(S_t,D_t)) + X_t \\
    D_{t+1} &\geq D_t(1-\delta) + G(S_t,D_t) + m X_t\\
    X_t &\in [0,\bar{X}] ~~ \forall t \\
    S_{t+1} & \geq 0, D_{t+1} \geq 0 \label{constraint}\\
    S_0 &= s_0, D_0 = d_0
\end{align}

For generality, we include an upper bound $\bar{X}$ on the allowable launch rate. If this never binds then the appropriate shadow value will simply be identically zero ($\gamma_{\bar{X}_t} \equiv 0$). The planner's Lagrangian is
\begin{align}
    \mathscr{L}(X,S,D,\lambda,\gamma) &= \sum_{t=0}^{\infty} \left( \frac{1}{1+r} \right)^t\bigg\{\pi S_t - F X_t + \lambda_{S_t} \left( S_t(1 - L(S_t,D_t)) + X_t - S_{t+1} \right)  \nonumber\\
    &+ \lambda_{D_t} \left( D_{t+1} - D_t(1-\delta) - G(S_t,D_t) - mX_t \right) \nonumber\\
    &+ \gamma_{X_t} X_t + \gamma_{\bar{X}_t}(\bar{X} - X_t) + \gamma_{S_t} S_{t+1} + \gamma_{D_t} D_{t+1}   \bigg\} \label{eqn:plan-lagrangian}
\end{align}

The first-order necessary conditions for an optimal launch path are, $\forall t$ up to $T$,
\begin{align}
     \mathscr{L}_{X_t} &= -F + \lambda_{S_t} - m \lambda_{D_t} + \gamma_{X_t} - \gamma_{\bar{X}_t} = 0\label{FOCi}\\
    \mathscr{L}_{S_{t+1}} &=  \frac{1}{1+r}  \{ \pi + \lambda_{S_{t+1}} (1 - L(S_{t+1},D_{t+1}) - S_{t+1} L_S(S_{t+1},D_{t+1})) \nonumber\\ &- \lambda_{D_{t+1}} G_S(S_{t+1},D_{t+1})\} + \gamma_{S_t} - \lambda_{S_t} = 0\label{FOCii}\\
    \mathscr{L}_{D_{t+1}} &=  \frac{1}{1+r}  \{ \lambda_{D_{t+1}} (\delta - 1 - G_D(S_{t+1},D_{t+1})) - \lambda_{S_{t+1}} S_{t+1} L_D(S_{t+1},D_{t+1})\} +\lambda_{D_t}+ \gamma_{D_t} = 0 \label{FOCiii}\\
    \mathscr{L}_{S_{T+1}} &=  \gamma_{S_T} - \lambda_{S_T} = 0 \label{FOCiv}\\
     \mathscr{L}_{D_{T+1}} &= \lambda_{D_T}+ \gamma_{D_T} = 0 \label{FOCv}
\end{align}
with complementary slackness and transversality conditions
\begin{align}
    \lambda_{St} \left( S_t(1 - L_t + X_t - S_{t+1} \right) & = 0  \\
    \lambda_{Dt} \left( D_{t+1} - D_t(1-\delta) - G_t - mX_t \right) & = 0  \\
    \gamma_{Xt} X_t &= 0, \\
    \gamma_{\bar{X}t}(\bar{X} - X_t) &= 0, \\
    \gamma_{St} S_{t+1} &= 0, \\
    \gamma_{Dt} D_{t+1} &= 0  \\
    \lim_{T \to \infty} \left( \frac{1}{1+r} \right)^{T} \lambda_{ST}S_{T+1} &= 0  \\ 
    \lim_{T \to \infty} - \left( \frac{1}{1+r} \right)^{T}\lambda_{DT}D_{T+1} &= 0.
    \label{debrisTVC}
\end{align}

In what follows we drop time subscripts to reduce notational clutter. Period $t$ values are shown with no subscript, period $t+1$ values are marked with a $'$ after the variable, and period $t-1$ values are marked with a $'$ before the variable e.g. $S_{t-1} \equiv~ 'S$, $S_t \equiv S$, $S_{t+1} \equiv S'$. By ($\ref{FOCi}$),
\begin{align}
    \lambda_S &= (1+r)(F + \frac{1}{1+r} m \lambda_D - \gamma_{X} + \gamma_{\bar{X}}) \label{WS_lagrangian}.
\end{align}

In the next period, this becomes
\begin{align}
    \lambda_S' &= (1+r)(F + \left(\frac{1}{1+r}\right) m \lambda_D' - \gamma_{X}' + \gamma_{\bar{X}}') \label{WS2_lagrangian}.
\end{align}

By ($\ref{FOCii}$) and ($\ref{FOCiii}$),
\begin{align}
    \lambda_S &= \pi + (1+r)\gamma_S + \frac{1}{1+r} \{\lambda_S' (1 - L(S',D') - S' L_S(S',D')) - \lambda_D' G_S(S',D')\}  \\ 
    \lambda_D &= \frac{1}{1+r} \{\lambda_D' (1 + G_D(S',D') - \delta) + \lambda_S' S' L_D(S',D') \} - (1+r)\gamma_D.
\end{align}
Using ($\ref{WS2_lagrangian}$),
\begin{align}
   \lambda_S &= \pi + (1+r)\gamma_S - F(L(S',D') + S' L_S(S',D') -1) -\frac{1}{1+r} \lambda_D G_S(S',D') \\ &-\frac{1}{1+r} m \lambda_D(L(S',D') + S' L_S(S',D') -1) + (L(S',D') + S' L_S(S',D') -1)(\gamma_{X}' - \gamma_{\bar{X}}') \\
    \lambda_D &= F S' L_D(S',D') + \frac{1}{1+r} \lambda_D' (1 + G_D(S',D') - \delta) + \frac{1}{1+r} m \lambda_DS' L_D(S',D') \\ & - \bigg ((1+r)\gamma_D + S' L_D(S',D') (\gamma_{X}' - \gamma_{\bar{X}}' \bigg)
\end{align}

Define
\begin{align}
    \alpha_1' &= \pi + (1 - L(S',D') - S' L_S(S',D')) F\\
    \alpha_2' &= S' L_D(S',D') F \\
    \Gamma_1' &= G_S(S',D') - m(1 - L(S',D') - S' L_S(S',D')) \\
    \Gamma_2' &= 1 - \delta + G_D(S',D') + m S' L_D(S',D')\\
    \kappa_1' &= (1+r)\gamma_S - (\gamma_{X'} - \gamma_{\bar{X'}})(1 - L(S',D') - S' L_S(S',D'))\\
    \kappa_2' &= (1+r)\gamma_D + S' L_D(S',D') (\gamma_{X'} - \gamma_{\bar{X'}}),
\end{align}
so that
\begin{align}
   \lambda_S &= \alpha_1' - \frac{1}{1+r} \lambda_D' \Gamma_1' + \kappa_1' \label{WS_simplified_lagrangian}\\
    \lambda_D &= \alpha_2' + \frac{1}{1+r} \lambda_S' \Gamma_2' - \kappa_2'. \label{WD_simplified_lagrangian}
\end{align}

Then,
\begin{equation}
    \lambda_D' = \frac{\lambda_D - \alpha_2' + \kappa_2'}{\frac{1}{1+r} \Gamma_2'}. \label{WD2_lagrangian}
\end{equation}
Substitute ($\ref{WS_lagrangian}$) and ($\ref{WD2_lagrangian}$) in ($\ref{WS_simplified_lagrangian}$) to get the following expression for $W_D(S,D)$
\begin{equation}
     \frac{\frac{1}{1+r}\{\Gamma_1'(\alpha_2' - \kappa_2') +  \Gamma_2'(\alpha_1' + \kappa_1') \} + \Gamma_2'(\gamma_{X}-\gamma_{\bar{X}}-F)}{\frac{1}{1+r} (\Gamma_1' + m \Gamma_2')}. \label{WD_elongated_lagrangian}
\end{equation}

Iterate $\ref{WD_elongated_lagrangian}$ to period $t+1$ and substitute into \ref{WD2_lagrangian} to obtain

\begin{equation}
     \lambda_D' = \frac{\frac{1}{1+r}\{\Gamma_1''(\alpha_2'' - \kappa_2'') +  \Gamma_2''(\alpha_1'' + \kappa_1'') \} + \Gamma_2''(\gamma_{X}'-\gamma_{\bar{X}}'-F)}{\frac{1}{1+r} (\Gamma_1'' + m \Gamma_2'')}.  \label{WD2_elongated_lagrangian}
\end{equation}

Use ($\ref{WD_elongated_lagrangian}$) and ($\ref{WD2_elongated_lagrangian}$) in ($\ref{WD_simplified_lagrangian}$) to get
\begin{equation}
       \alpha_1' = m(\alpha_2' - \kappa_2') - \kappa_1' + \frac{1}{\frac{1}{1+r}} (\gamma_{\bar{X}} - \gamma_X + F) + \frac{\Gamma_1' + m \Gamma_2'}{\Gamma_1'' + m \Gamma_2''} \bigg ( \Gamma_1'' \frac{1}{1+r} (\alpha_2'' -  \kappa_2'') + \Gamma_2''(\frac{1}{1+r} (\alpha_1'' +  \kappa_1'') - F + \gamma_{X}' - \gamma_{\bar{X}}')\bigg). \label{approach2_alpha1'}
\end{equation}
Evaluate ($\ref{approach2_alpha1'}$) in the previous time period as:
\begin{align}
\hspace*{-2cm}
    \alpha_1 = m(\alpha_2 - \kappa_2) - \kappa_1 + \frac{1}{\frac{1}{1+r}} ('\gamma_{\bar{X}} - '\gamma_{X} + F) + \frac{\Gamma_1 + m \Gamma_2}{\Gamma_1' + m \Gamma_2'} \bigg ( \Gamma_1' \frac{1}{1+r} (\alpha_2' -  \kappa_2') + \Gamma_2'(\frac{1}{1+r} (\alpha_1' +  \kappa_1') - F + \gamma_{X} - \gamma_{\bar{X}})\bigg) .
\end{align}
Subtract $F (\frac{1}{\frac{1}{1+r}} + L')$ from both sides and add $F (L + S L_S)$ to both sides to obtain
\begin{align}
\hspace*{-2cm}
    \pi - r F - F L(S',D') &= F (L(S,D) + S L_S(S,D) - L(S',D')) + m(\alpha_2 - \kappa_2) - \kappa_1 + \frac{1}{\frac{1}{1+r}} ('\gamma_{\bar{X}} - '\gamma_{X})  + \frac{\Gamma_1 + m \Gamma_2}{\Gamma_1' + m \Gamma_2'} \nonumber\\ 
\hspace*{-2cm}
    &\bigg( \Gamma_1' \frac{1}{1+r} (\alpha_2' -  \kappa_2') + \Gamma_2'(\frac{1}{1+r} (\alpha_1' +  \kappa_1') - F + \gamma_{X} - \gamma_{\bar{X}})\bigg)\\
\hspace*{-2cm}
    \implies \xi(S',D') &= \underbrace{L_S(S,D) S F + \left( L(S,D) -L(S',D')\right)F}_{\text{\shortstack{Congestion channel}}} +\underbrace{\frac{\Gamma_1 + m \Gamma_2}{\Gamma_1' + m \Gamma_2'}\Gamma_2' (\frac{1}{1+r} \alpha_1' - F)}_{\text{\shortstack{Pollution persistence channel}}} + \underbrace{\frac{1}{1+r} \frac{\Gamma_1 + m \Gamma_2}{\Gamma_1' + m \Gamma_2'}\Gamma_1' \alpha_2'}_{\text{\shortstack{Pollution hazard channel}}} \nonumber \\
\hspace*{-2cm}
    &\underbrace{+ m\alpha_2}_{\text{\shortstack{Pollution\\hazard\\channel}}}+ \underbrace{\frac{\Gamma_1 + m \Gamma_2}{\Gamma_1' + m \Gamma_2'}\bigg(\Gamma_2'(\frac{1}{1+r}\kappa_1' + \gamma_{X} - \gamma_{\bar{X}}) - \frac{1}{1+r} \Gamma_1' \kappa_2'\bigg) - (m\kappa_2 + \kappa_1) + \frac{1}{\frac{1}{1+r}}(\gamma_{\bar{'X}} - \gamma_{'X})}_{\text{\shortstack{Adjustments for prior or upcoming corner solutions}}}. \label{eqn:general-mec}
\end{align}

Along an interior launch path, the MEC $\xi(S',D')$ reduces to
\begin{align}
\hspace*{-1.5cm}
    \xi(S',D') &= L_S(S,D) S F + \left( L(S,D) -L(S',D')\right)F +\frac{\Gamma_1 + m \Gamma_2}{\Gamma_1' + m \Gamma_2'}\Gamma_2' (\frac{1}{1+r} \alpha_1' - F) + \frac{1}{1+r} \frac{\Gamma_1 + m \Gamma_2}{\Gamma_1' + m \Gamma_2'}\Gamma_1' \alpha_2' + m\alpha_2,
\end{align}

and in an interior steady state the MEC further reduces to
\begin{align}
    \xi(S,D) &= L_S(S,D) S F + \Gamma_2 (\frac{1}{1+r} \alpha_1 - F) + \frac{1}{1+r} (\Gamma_1 + m) \alpha_2.
\end{align}

\subsection{The collision probability and new fragment formation functions}
\label{appendix:kinetic_gas_model}

In this section we derive the functional forms of the collision probability and new fragment functions, discuss the physical assumptions they encode, and describe our process for calibrating the physical model in more detail.

For numerical simulations, we model the probability that objects of type $j$ are struck by objects of type $k$ as
\begin{align}
\label{lossprobgeneric}
p_{jk}(k_t) = 1 - e^{-\alpha_{jk} k_t},
\end{align}

where $\alpha_{jk} > 0$ is a physical parameter (``intrinsic collision probability'') reflecting the relative mean sizes, speeds, and inclinations of the object types (see \cite{letizia2016space} for a derivation of the physical content of $\alpha_{jk}$). The probability a satellite is destroyed is the sum of the probabilities it is struck by debris and by other satellites, adjusted for the probability it is struck by both. For satellite-satellite and satellite-debris collisions, equation \ref{lossprobgeneric} gives us
\begin{align}
\label{lossformgeneric}
L(S,D) &= p_{SS}(S) + p_{SD}(D) - p_{SS}(S)p_{SD}(D) \\
&= (1 - e^{-\alpha_{SS} S}) + (1 - e^{-\alpha_{SD} D}) - (1 - e^{-\alpha_{SS} S})(1 - e^{-\alpha_{SD} D}) \nonumber \\
\label{lossform}
\implies L(S,D) &= 1 - e^{-\alpha_{SS}S -\alpha_{SD}D} .
\end{align}

We write the new fragment formation function as
\begin{align}
G(S,D) = F_{SD}p_{SD}(D) + F_{SS}p_{SS}(S) + F_{DD}p_{DD}(D),
\end{align}

where $F_{jk}$ is the number of fragments produced in a collision between objects of type $j$ and $k$. Letting $F_{SS} = \beta_{SS}S$, $F_{SD} = \beta_{SD} S$, and $F_{DD} = \beta_{DD} D$ where $\beta_{jk} > 0$ is a physical parameter reflecting the physical compositions and masses of the colliding objects, and using the forms in equation \ref{lossprobgeneric}, we obtain
\begin{align}
\label{newfragform}
G(S,D) = \beta_SS (1 - e^{-\alpha_{SS} S}) S + \beta_{SD} (1 - e^{-\alpha_{SD} D}) S + \beta_{DD} (1 - e^{-\alpha_{DD} D}) D .
\end{align}

The form in equation \ref{lossform} is convenient as it allows us to solve explicitly for the open access launch rate and is easy to manipulate. Similar forms have been used in engineering studies of the orbital debris environment, and are currently used by the European Space agency in developing indices to study the long-term evolution of the orbital environment \citep{letizia2016space, letiziaetal2017, letiziaetal2018}.

To derive equation \ref{lossform}, we consider balls (satellites and debris) being placed into bins (the set of all possible orbital paths within the shell of interest). The probability of a specific satellite being struck by another object is then equivalent to the probability that a randomly-placed ball ends up in a bin containing the specific ball we are focusing on. This is a version of the ``pigeonhole principle'', used in \cite{bealetal2020} to derive a similar form for satellite-satellite collisions. 

Suppose we have $b$ equally-sized bins and $n+1$ balls in total, where $b \geq n+1$. Without loss of generality, we label the ball we are interested in as $i$. We will first place $i$ into an arbitrary bin, and then drop the remaining $N$ balls into the $b$ bins with equal probability over bins. The probability a ball is dropped into a given bin is $\frac{1}{b}$, and the probability a ball is not dropped into a given bin is then $\frac{b-1}{b} = 1 - \frac{1}{b}$. As we drop the remaining $n$ balls, the probability that none of the balls is dropped in the same bin containing $j$ is 
\begin{equation}
Pr(\text{no collision with $i$}) = \left( 1 - \frac{1}{b} \right)^n
\end{equation}

Consequently, the probability that any of the $n$ balls are dropped into $i$'s bin is
\begin{equation}
Pr(\text{collision with $i$}) = 1 - \left( 1 - \frac{1}{b} \right)^n .
\end{equation}

Now suppose we are interested in the probability that members of a collection of $j$ balls, $1 \leq j < b$, end up in a bin with one of the remaining $n+1-j$ balls. The probability that any of the remaining balls end up in a bin with any of the $j$ balls we are interested in is then
\begin{equation}
Pr(\text{collision with $i$}) = 1 - \left( 1 - \frac{j}{b} \right)^{n+1-j} .
\end{equation}

As the number of bins and balls grow large ($\lim_{b,n \to \infty}$), we obtain
\begin{equation}
Pr(\text{collision with $i$}) = 1 - e^{-j} .
\end{equation}

Though neither the number of objects in orbits nor the possible positions they could occupy is infinite, the negative natural exponential form is likely a reasonable approximation. If we suppose that we have two types of balls $j$ and $k$ of different sizes and bins the size of the smallest type of ball, we get that the probability a ball of type $k$ is dropped into in a bin with a ball of type $j$ as
\begin{align}
Pr(\text{$k$--$j$ collision}) &= 1 - \left( 1 - \frac{\alpha_{jk}k}{b} \right)^{n+1-k} \\
\implies \lim_{b,n \to \infty } Pr(\text{$k$--$j$ collision}) &= 1 - e^{-\alpha_{jk} k},
\end{align}

which is the form in equation \ref{lossprobgeneric}, where $\alpha_{jk}$ is a nonnegative parameter indexing the relative sizes of objects $j$ and $k$. In the orbital context, $\alpha_{jk}$ reflects not only the sizes of the objects but also their relative speeds and inclinations. From here we obtain the form of $L$ by applying standard rules of probability to satellite-satellite and satellite-debris collisions. Equation \ref{newfragform} follows from the form of $L$.

This ``kinetic gas-like'' approximation is used extensively in the space debris modeling literature as a tractable approximation of results from more complex and computationally-intensive orbital mechanics simulators. It is most suitable for long-term modeling studies with ``large'' (relative to the timescale of orbital interactions) time steps. As described in \cite{letizia2016space}, this approximation is equivalent to modeling collisions as a Poisson process. The Poisson assumption that the number of events occurring in non-overlapping time intervals are independent is equivalent to assuming that objects move randomly throughout the shell volume. This assumption is clearly not true, leading to our regularization approach described below. The assumption that the probability of an event is proportional to the length of the interval implies that fragment clouds are dispersed enough, and contain enough fragments, to be considered a continuum. Since our model is solved at annual timesteps while debris clouds evolve at much smaller timescales, this assumption is reasonable for our purposes.

\subsection{Modeling debris growth over the next century}
\label{appendix:debris_unbounded}

As we note in the main text, truly ``unbounded'' growth is unphysical, as collisional activity will reduce the fragments to smaller sizes and objects in LEO will eventually decay due to drag, solar radiation pressure, and other orbital perturbations. However, we follow the existing engineering literature on source-sink evolutionary models of the debris environment in allowing unbounded growth over the next century \citep{talent1992analytic, lewis2009fast, lifson2022many}. An example using empirical data from a collision and the size-energy scaling law may help illustrate the underlying reasoning for this modeling choice.

Consider a collision between two large intact bodies, e.g. an event like the Iridium-Cosmos collision on February 10, 2009. Iridium 33 was an operational US communications satellite ($S$) while Cosmos 2251 was defunct Russian communications satellite ($D$). The table below from \citet{kelso2009analysis} shows the relevant size and mass characteristics of the initial objects and resulting fragments. 

\begin{table}[htbp]
  \centering
  \caption{``Table 1. Pre-Collision Satellite Characteristics.'' from \citet{kelso2009analysis}}
  \begin{tabular}{ccccc}
    \hline
    Satellite & Number of Pieces & Total Volume (m$^3$) & Dry Mass (kg) & Inclination (deg) \\
    \hline
    Iridium 33 & 386 & 3.388 & 556 & 86 \\
    Cosmos 2251 & 927 & 7.841 & 900 & 74 \\
    \hline
  \end{tabular}
\end{table}

The average radii for fragments from Iridium and Cosmos were around 12.8 cm and 12.6 cm, with average masses around 1.44 kg and 0.971 kg. These figures imply that the tracked fragments larger than 10 cm radius account for most of the initial body masses.\footnote{10 cm is also the lower detection limit for sensor systems, raising concerns about censoring. The mass accounting suggests censoring may not be quantitatively large.}

The relation between a uniform sphere's kinetic energy and mass, given density $\rho$ and velocity $v$, is
\begin{equation}
    KE(r) = \frac{1}{2} \underbrace{ \rho \left( \frac{4}{3} \pi r^3 \right)}_{\text{mass = density$\times$volume}} v^2.
\end{equation}
Suppose a fragment of around 10 cm radius is a uniform aluminum sphere---a common assumption in debris modeling given the prevalence of aluminum in satellite construction, e.g. \citet{letizia2016space}. Aluminum has a mass of around 2.7 g/cm$^3$, giving a volume of 4188 cm$^3$ and mass of around 11 kg. Typical objects in low-Earth orbit have velocities on the order of 10 km/s \citep{lifson2022many, d2023novel}.\footnote{Velocity in orbit is linked with altitude---accelerating or decelerating along its forward direction raises and lowers altitude, respectively.} Such a fragment will therefore have a kinetic energy of roughly 550 megajoules , or approximately 131 kg of TNT (energy equivalent of 1 kg of TNT is 4.184 megajoules). This is in the category of ``hypervelocity'' impacts that can shatter the intact object \citep{esa_hypervelocity1}. If the object is like Iridium or Cosmos---not-atypical LEO satellites---it may produce hundreds of fragments.

Since mass scales cubically with object radius, a reduction in average fragment size to 1 cm radius reduces the mass to 0.011 kg, producing an impact energy of 0.55 megajoules---comparable to the force of a hand grenade \citep{esa_hypervelocity2}. Even if it takes tens of collisions with fragments of 1-10 cm radius to overcome shielding on a large intact object, the resulting tens or hundreds of fragments will ensure net growth. To the extent that these objects move in debris ``fields''---which may occur systematically due to orbital mechanics factors, particularly when a larger body is struck by a smaller one, e.g. \citet{oltrogge2022iaa, oltrogge2022comparison, pardini2023short}--- their lethal effects at these and even smaller sizes may be amplified.

Suppose we take 1 cm to be a conservative ``lethal size limit''. How long will it take for collisional activity to reduce a fragment below this limit? Suppose the average cumulative annual collision probability for an arbitrary debris fragment is 25\%---perhaps a high estimate, but again erring on the side of caution. That fragment will go roughly 4 years between collisions. If fragments are reduced to roughly 1 cm radius after only two collisions, it would take about 8 years for that debris fragment and its children to be rendered nonlethal. At 1\% collision probability, the fragment's lethal lifetime is around 200 years.

At 575 km altitude, a large intact object has a residence time (i.e. time before it falls back to Earth due to drag) on the order of 10 years, and a 10 cm fragment has a residence time on the order of a year, for an upper bound on lethal lifetime of around 11 years. At 775 km altitude, the residence times are around 190 years for an intact object and 10 years for a fragment, for an upper bound on lethal lifetime of around 200 years. During their residence times the objects slowly drift downwards, entering lower shells. Most satellites are currently near or above 575 km altitude. Since plausible lethal lifetimes are on the order of relevant residence times, debris are likely to spend most of their lives at lethal sizes. Given a sufficiently large amount of mass at currently-popular altitudes (e.g. 100,000 satellites at 250 kg each spread over 550-800 km altitude), it seems reasonable to consider potential growth to ``unbounded'' levels over the next century.

\subsection{Open access with a finite horizon}
\label{appendix:finite-horizon-orbital-capacity}

We employ an infinite-horizon modeling approach in the general model. However, one may reasonably wonder whether our conclusions regarding the open-access equilibrium are sensitive to this point. In this section we show that a finite-horizon problem with terminal period $T$ (where it either becomes prohibitively costly to use the volume or Kessler Syndrome occurs or both) produces the same equilibrium condition. 

Suppose there exists a final period, $T$, such that the potential launchers will all exit the market. We are agnostic as to why this may be the case, except to note that if such a period exists, it must be that there are no profits to be gained from launching after that period. In the final period, the launcher's value becomes
\begin{equation}
    V_{iT}(S_T,D_T,X_T) = \max_{x_{iT} \in \{0,1\}} \{ (1-x_{iT}) \frac{1}{1+r} V_{iT+1}(S_{T+1},D_{T+1},0) + x_{iT} \left[ \frac{1}{1+r} Q(S_{T+1},D_{T+1}) - F \right]  \}.
\end{equation}

There are two possible cases here for the value of launching in the final period, $\frac{1}{1+r} Q(S_{T+1},D_{T+1}) - F$:
\begin{enumerate}
    \item $\frac{1}{1+r} Q(S_{T+1},D_{T+1}) - F = 0$. In this case the potential launchers are indifferent between launching in the final period or not launching. By backwards induction the equilibrium path up to period $T$ will match the one derived in the general model in the main text, with equation \eqref{eqn:zeroprofit} being the equilibrium condition.
    \item $\frac{1}{1+r} Q(S_{T+1},D_{T+1}) - F < 0$. In this case, firms would prefer not to launch. Optimization by individual launchers therefore implies $V_{iT}(S_T,D_T,X_T) = 0$. This matches equation \eqref{eqn:launcher_zero_profit}, which yields \eqref{eqn:zeroprofit} after some algebra. So again by backwards induction the equilibrium path up to period $T$ will match the one derived in the general model in the main text.
\end{enumerate}

Indeed, it is possible to go one step further: the existence of such a terminal period (where $X_{t} = 0 ~ \forall t \geq T$) is possible if and only if Kessler Syndrome occurs along the equilibrium path.

\begin{prop}
\label{prop:finite-horizon-zero-launch}
    A terminal period $T$ where $\hat{X}_{t} = 0 ~ \forall t \geq T$ can exist for an open-access equilibrium path $\{\hat{X}_{t}\}_t$ if and only if Kessler Syndrome occurs (i.e. $\lim_{t \to \infty} D_t = \infty$) along the open-access equilibrium path.
\end{prop}

\begin{proof}
    The proposition asserts that 
    \begin{equation}
        X_{t} = 0 ~ \forall t \geq T \iff \lim_{t \to \infty} D_t = \infty
    \end{equation}
    
    We first show the $\impliedby$ direction, then the $\implies$ direction. \\

    \textbf{The ``only if'' direction, $X_{t} = 0 ~ \forall t \geq T \impliedby \lim_{t \to \infty} D_t = \infty $:} If $\lim_{t \to \infty} D_t = \infty$, then there is some period $\bar{t}$ such that $D_{t} > D_{\bar{t}}$ for all $t > \bar{t}$. From the law of motion for $D$ and our assumption that $\lim_{t \to \infty} D_t = \infty$, we can see that $D_t$ must be monotonically increasing after $\bar{t}$. So there must exist a period $T \geq \bar{t}$ such that $L(S_t, D_t)F > \pi - rF $ for all $t \geq T$, i.e. where it becomes unprofitable to launch one more satellite at that or any future period. Thus, $X_{t} = 0 ~ \forall t \geq T$. This completes the $\impliedby$ direction. \\

    \textbf{The ``if'' direction, $X_{t} = 0 ~ \forall t \geq T \implies \lim_{t \to \infty} D_t = \infty$:} If $X_t = 0$ for all $t \geq T$ then it must be the case that $\frac{1}{1+r} Q(S_{t+1},D_{t+1}) - F < 0$ for all $t \geq T$, else some firm would find it profitable to launch.  Note that it must be unprofitable to launch at $t$ given that there are \textit{no} launches occurring at $t$. \\
    
    To be explicit in the next steps, we write the satellite and debris stocks with the previous-period aggregate launch rate $X_t$ shown explicitly as an argument, i.e. writing $S_{t+1}(X_t)$ and $D_{t+1}(X_t)$. Along a path $\{ S_{t}(0) \}_{t>T}^{\infty}$, clearly $S_{t+1}(0) \leq S_{t}(0)$. Now, $\frac{1}{1+r} Q(S_{t+1}(0),D_{t+1}(0)) - F < 0$ for all $t \geq T$ implies that $L(S_{t+1}(0),D_{t+1}(0))F > \pi - rF$ for all $t \geq T$. Monotonicity of $L$ and $S_{t+1}(0) \leq S_{t}(0)$ then imply that $D_{t+1}(0) \geq D_t(0)$. \\

    If there exists a threshold $D^K$ such that $\lim_{t \to \infty} D_t(X_t) = \infty$ when $D > D^K$ for any $X_t$, then there are only two possible cases: either $\lim_{t \to \infty} D_t (0) < D^K$, or $\lim_{t \to \infty} D_t (0) \geq D^K$. The first case is a contradiction when $G$ is strictly convex increasing, as each increase in $D_{t+1} - D_t$ must be larger than $D_{t} - D_{t-1}$ so eventually $D_t$ must exceed $D^K$. Only the second case is consistent with the general physical model. This completes the $\implies$ direction.
\end{proof}

Finally, how large is the volume available to be filled? Recent analyses estimate the maximum capacity consistent with stable orbital populations (i.e. no Kessler Syndrome) between 200-900km altitude to be on the order of 1.8 million active satellites, assuming no debris \citep{lifson2022many}. Over the next few decades, the total number of objects slated for launch is expected to be on the order of 80,000 satellites \citep{aerospace_futuresats}. It is unclear whether there is sufficient demand to support hundreds of thousands of satellites, let alone over a million. While we do not think the maximum capacity described in the engineering literature will be realized due to both the externalities described here and in the economic literature and the aforementioned demand limitations, the large capacity available makes the issue seem less one of filling the volume with satellites or debris than one of operating in the volume becoming too costly due to risk.

\section{Calibration details}
\label{appendix:calibration}

\subsection{Data}
We calibrate the economic parameters of our model using data collected by The Space Report \citep{tsr2021} on the annual revenues accruing to each sector of the space economy from 2006-2019. These data have been used in other economic analyses of space and orbit use \citep{wienzierl2018, raoetal2020, crane2020measuring, raoletizia2021}. The data are not ideal for our purpose as they are aggregates covering the entire space sector, but more granular datasets describing specific LEO satellite operators' revenues and costs are not available. To focus on revenues and costs relevant to LEO satellite operators, we use only the variables which are plausibly attributable to LEO satellite activities. We calculate total LEO satellite operator revenues as the sum of the ``Satellite communications'' and ``Earth observation'' variables, and total LEO satellite operator costs as the sum of the ``Ground stations and equipment'', ``Space Situational Awareness'' (SSA),  ``Insurance premiums'', ``Commercial satellite launch'', and ``Commercial satellite manufacturing'' variables. We discard variables representing revenues to the direct-to-home television, GNT (Geolocation, Navigation, and Timing), and satellite radio sectors, as these are provided by satellites in higher orbits beyond LEO. We also exclude suborbital commercial human spaceflight deposits as they are by definition for transit to regions below orbital altitudes (e.g. 50-80 km above mean sea level). Since our data is recorded annually, we set the period length to 1 year.  We display the calculated variables in table \ref{tab:econdata}. Note that these are \emph{not} the revenues and costs accruing specifically to LEO operators---a distinction not possible given our data. Rather, these variables represent a superset of LEO operator revenues and costs, as they necessarily include some geostationary satellites. We describe our strategy to account for this issue during calibration in Appendix \ref{appendix:econ-calibration-strategy}. 
\begin{table}[!htbp]
  \centering
  \caption{Economic data. Figures are in nominal billion USD. Data from \citet{tsr2021} and authors' calculations.}
  \label{tab:econdata}
    \begin{tabular}{ccc}
      \hline \hline
      Year & Maximum total revenues & Maximum total costs \\
           & attributable to all operators & attributable to all operators \\
           & potentially using LEO & potentially using LEO \\
      \hline
      $2006$ & $13.800$ & $80.840$ \\
      $2007$ & $16.368$ & $92.956$ \\
      $2008$ & $18.104$ & $85.371$ \\
      $2009$ & $18.695$ & $69.270$ \\
      $2010$ & $19.570$ & $68.460$ \\
      $2011$ & $21.424$ & $83.853$ \\
      $2012$ & $22.747$ & $93.779$ \\
      $2013$ & $23.683$ & $108.199$ \\
      $2014$ & $24.002$ & $127.567$ \\
      $2015$ & $25.884$ & $87.222$ \\
      $2016$ & $26.087$ & $89.201$ \\
      $2017$ & $26.545$ & $95.857$ \\
      $2018$ & $28.420$ & $99.930$ \\
      $2019$ & $27.320$ & $119.160$ \\
      \hline
    \end{tabular}
\end{table}

We calibrate physical parameters of our model using a kinetic gas approximation of orbital mechanics and data from DISCOS \citep{letiziaetal2017, DISCOS}. These data describe the launch traffic, active satellites, and tracked debris objects (i.e larger than 10 cm diameter) in the 600-650 km shell over the 2006-2020 period. These data aggregate over different types of operators (e.g. commercial operators, civil government operators, defense operators). We display these data in table \ref{tab:phys_data}, along with the collision probability calculated from the kinetic gas approximation assuming satellite operators avoid 99\% of all collisions between satellites and 95\% of all collisions between satellites and tracked debris. Letting the avoidance success rates be $\kappa_{SS}$ and $\kappa_{SD}$, the probability of an unavoidable collision becomes 
\begin{equation}
    L(S,D) = (1-\kappa_{SS})(1 - e^{-\alpha_{SS}S}) + (1 - \kappa_{SD})(1 - e^{-\alpha_{SD}D}) - (1-\kappa_{SS})(1 - \kappa_{SD})(1 - e^{-\alpha_{SS}S})(1 - e^{-\alpha_{SD}D}).
\end{equation}

Many ostensibly-non-commercial satellites are operated as joint ventures with commercial enterprises and many commercial satellite operators serve primarily civil government or defense customers, so we do not separate the satellite data by operator type. Further, since all satellites contribute to debris and collision probability regardless of their operator type, non-commercial operators' satellites ought to be included in the state vector. Non-commercial operators may also contribute to the observed ``occupancy elasticity'' (described precisely in the following section), further complicating efforts to properly disentangle payoffs to different operator types from the available data.

The DISCOS physical data also provide object characteristics such as mass and cross-sectional area, which are necessary for the kinetic gas approximation. We describe the details of the kinetic gas approximation of orbital mechanics in Appendix \ref{appendix:physical_model_calibration}.

\begin{table}[!htbp]
  \centering
  \caption{Orbital traffic in the 600-650 km shell. Collision probability is rounded. Data from \citet{DISCOS} and authors' calculations.}
  \begin{tabular}{ccccc}
    \hline \hline
    Year & Satellites launched & Active satellites & Tracked debris & Collision probability \\
    & & satellites & & \\
    \hline
    $2006$ & $15$ & $25$ & $211$ & $2.95 \times 10^{-6}$ \\
    $2007$ & $84$ & $31$ & $275$ & $3.84 \times 10^{-6}$ \\
    $2008$ & $168$ & $47$ & $273$ & $3.85 \times 10^{-6}$ \\
    $2009$ & $72$ & $43$ & $393$ & $5.48 \times 10^{-6}$ \\
    $2010$ & $156$ & $53$ & $444$ & $6.20 \times 10^{-6}$ \\
    $2011$ & $30$ & $56$ & $411$ & $5.76 \times 10^{-6}$ \\
    $2012$ & $73$ & $53$ & $429$ & $6.00 \times 10^{-6}$ \\
    $2013$ & $213$ & $64$ & $454$ & $6.37 \times 10^{-6}$ \\
    $2014$ & $261$ & $97$ & $484$ & $6.87 \times 10^{-6}$ \\
    $2015$ & $175$ & $122$ & $495$ & $7.09 \times 10^{-6}$ \\
    $2016$ & $15$ & $114$ & $494$ & $7.05 \times 10^{-6}$ \\
    $2017$ & $26$ & $122$ & $525$ & $7.49 \times 10^{-6}$ \\
    $2018$ & $36$ & $139$ & $506$ & $7.28 \times 10^{-6}$ \\
    $2019$ & $33$ & $155$ & $543$ & $7.83 \times 10^{-6}$ \\
    $2020$ & $9$ & $158$ & $626$ & $8.97 \times 10^{-6}$ \\
    \hline
  \end{tabular}
  \label{tab:phys_data}
\end{table}

\subsection{Economic calibration}
\label{appendix:econ-calibration-strategy}

To calibrate our economic model, we make three modifications to the open-access equilibrium condition in equation \eqref{eqn:zeroprofit-rate}. First, we allow the per-period satellite payoff and cost to vary over time, i.e $\pi \to \pi_t$ and $F \to F_t$. This changes the equilibrium condition to 
\begin{align}
\label{OAeqm-tv}
\pi_{t+1} &= (1+r)F_t - (1 - L(S_{t+1},D_{t+1}))F_{t+1} \\
\implies L(S_{t+1},D_{t+1})  &= 1 + \frac{\pi_{t+1}}{F_{t+1}} - (1+r)\frac{F_t}{F_{t+1}}.
\label{OAeqm-tv-ratio}
\end{align}

This form is similar to the one described in equation \eqref{eqn:zeroprofit-rate} but for the time subscripts and term $1 - (1+r)\frac{F_t}{F_{t+1}}$. This term represents capital gains accruing to a period $t$ launcher from increases in the cost of building and launching a satellite in period $t+1$. We abstract from operators' expectations over economic variables and assume they perfectly forecast all $t+1$ objects.

Second, we allow the per-period satellite payoff to depend on the current stock of satellites in orbit, i.e $\pi_t \to p_t(S_t)$. We use a constant elasticity form with exponential factor productivity growth, $p_t(S_t) = \pi e^{at} (1 + \eta) S_t^{\eta}$, where $\eta$ is the ``orbital occupancy elasticity of per-period satellite payoffs''. We assume that the downstream market for satellite outputs is competitive such that operators do not internalize $\pdv{p_t}{S_t}$. The equilibrium condition becomes 
\begin{align}
L(S_{t+1},D_{t+1}) &= 1 + \frac{p_{t+1}(S_{t+1})}{F_{t+1}} - (1+r)\frac{F_t}{F_{t+1}}.
\end{align}

Third, we incorporate exogenous limited satellite lifespans to allow for natural depreciation and replacement of satellites. Specifically, we assume each satellite is replaced with probability $\mu$ each period. We calibrate this value explicitly to simulate object stocks (described in the Appendix \ref{appendix:physical_model_calibration}); for now, we leave this to be adjusted in the regression-based calibration approach described below. The final equilibrium condition for our simulations is
\begin{align}
L(S_{t+1},D_{t+1}) &= 1 + \frac{1}{1-\mu}\frac{p_{t+1}(S_{t+1})}{F_{t+1}} - \frac{1+r}{1-\mu}\frac{F_t}{F_{t+1}}.
\end{align}

To simulate future periods under different returns growth rate and occupancy elasticity assumptions, we estimate the growth rate of total LEO satellite operator costs. We estimate
\begin{align}
    \log(F_t) &= \eta_0^{F} + \eta_1^{F} t + \nu^{F}_t,
    \label{reg:growth_costs}
\end{align}

where $\log(F_t)$ is the natural log of total LEO satellite operator costs, $t$ is the year, the growth rate (the object of interest) is $g = \exp(\eta_1^{F}) - 1$, and the regression error is $\nu^{F}_t$. The estimated growth rate is roughly 2.5\%, which is consistent with \cite{crane2020measuring}.

There are two final steps to our procedure: ensuring consistency between the occupancy elasticity and factor productivity parameters, and accounting for unobserved variables. To ensure consistency between the assumed elasticity and implied orbital slot factor productivity and match the final observed value of LEO-using sector revenues (\$27.32b in 2019, see table \ref{tab:econdata}), we calibrate the factor productivity term $\pi$ in equation \ref{eqn:time-varying-payoff}. Specifically, letting $K$ be the observed value to match for each assumed elasticity value $\eta_j$, setting $t = 0$ and $S_0$ to the shell-specific initial condition ($S_0 = 158$), the factor productivity term $\pi_j$ satisfies
\begin{equation}
    \pi_j = \exp(\log(K) - \log(1 + \eta) + \eta \log(S_0)).
\end{equation}
 
Finally, as mentioned in the previous section, using maximum total sector revenues and costs directly from the data in table \ref{tab:econdata} as though the data reflects only operators in the 600-650 km shell is challenging for two reasons. First, the data in table \ref{tab:econdata} cover all satellite operators---our variable selection step is the only thing restricting the set of operators included in the data. Even if we were successful in removing all operators outside of LEO through variable selection when calculating total LEO operator revenues and costs, the revenue and cost variables will still include operators outside the 600-650 km shell. The data aggregation implies an unobservable ``shell-share'' coefficient, $s \in [0,1)$, scaling observed aggregate revenues and costs to reflect only the portion attributable to satellites in the 600-650 km shell. Second, theory predicts that the discount rate used by operators is a critical parameter in determining LEO use, but this parameter is unobserved.

Fortunately, equation \ref{OAeqm-tv-ratio} offers a way to address both challenges. Letting $\pi_t$ be the total LEO satellite operator revenues and $F_t$ be the total LEO satellite operator costs from table \ref{tab:econdata}, and $L_t$ be the collision probability shown in table \ref{tab:phys_data}, we estimate the following regression on data from 2006-2019:
\begin{align}
    L_t &= \gamma_0 + \gamma_1 \frac{\pi_{t}}{F_{t}} + \gamma_2 \frac{F_{t-1}}{F_{t}} + e_t.
    \label{eqn:eqm_cond_calibration}
\end{align}

The estimated ``adjustment coefficients'' $(\gamma_0, \gamma_1, \gamma_2)$ reflect the shell-share $s$, the discount rate $r$, as well as the satellite turnover $\mu$ (though they are not separately identified). We use the adjustment coefficients to simulate the model in future periods given projected growth in $\pi_t$ and $F_t$. If the shell-share coefficients (labeled $s$ in the preceding discussion) are common to revenues and costs and time-invariant (or ``close'' and ``slowly-varying''), they will (almost) cancel out of the ratios we use in equation \eqref{eqn:eqm_cond_calibration} and our estimated adjustment coefficients would only reflect satellite turnover and discounting.\footnote{This is not the only interpretation of our estimates---as described in \cite{raoetal2020}, the adjustment coefficients may also reflect unmodeled frictions in satellite launching and operation.}

\subsection{Physical calibration}
\label{appendix:physical_model_calibration}

Here we describe key equations and the ridge regression approach to correcting for non-random object paths. Readers interested in detailed explanations of the physics-based elements of our calibration approach, including derivations and validation, are referred to \cite{letizia2016space}.

We require physically-appropriate values for the following parameters: $\delta$, $\mu$, $\alpha_{SS}$, $\alpha_{SD}$, $\alpha_{DD}$, $\beta_{SS}$, $\beta_{SD}$, $\beta_{DD}$. Calibrating $\delta$ and $\mu$ (the mean debris decay rate and mean satellite active lifetime time) are the most straightforward. We take data from ESA regarding the residence time $\delta_{r}$ of debris objects and lifetime of active satellites $\mu_r$ at different altitudes \citep{DISCOS}. We set the decay rate for debris objects as $\delta = \min\{ 1 - \delta_r^{-1}, 1\}$ and the natural turnover rate for satellites as as $\mu = \min\{ 1 - \mu_r^{-1}, 1\}$. For both parameters we calculate share-weighted averages across object types within the category to reflect the effects of heterogeneous object dimensions, e.g. $\delta$ reflects the weighted average of decay times for rocket bodies, fragments, and intact derelict objects.\footnote{Cross-sectional area and mass are key determinants of orbital residence times. Both can vary significantly within object classes.} We calculate the share-weighted decay rate in the 600-650 km shell is roughly 7\% every year.\footnote{At these altitudes, the decay rates from higher shells rapidly approach zero. For the 650-700 km shell, the share-weighted average decay rate is roughly 5\%, and for the 700-750 km shell the decay rate is 3\%. We therefore neglect objects entering the 600-650 km shell from higher altitudes as they are unlikely to significantly change our results.} The share-weighted average active LEO satellite lifetime is roughly $6.71$ years. This implies roughly $15\%$ of active satellites in LEO turn over every year on average, i.e the fraction remaining is $1 - \mu = 0.85$.

We calibrate the parameters of $L$ and $G$ in two steps. First, we compute the collision probability and new fragment formation parameters using a kinetic gas approximation similar to the one used in \cite{letiziaetal2017} and \cite{letiziaetal2018} as well as analytical fragmentation formulas from \cite{krisko2011proper} and \cite{letizia2016space} calibrated to the NASA standard breakup model. These formulas require data on object mass and cross-sectional area, which we obtain from DISCOS. The DISCOS parameters describe average values across different types of active satellites and debris objects, so we compute share-weighted averages for active satellites and debris objects. The kinetic gas approximation implies that objects within the shell are moving randomly, leading to our next step. Second, to adjust for the non-random motion of objects in the shell, we regularize the expected fragmentation components of $G$ by estimating a ridge regression on the debris law of motion using data in table \ref{tab:phys_data} and the analytically-computed parameter values. We also use this second step to jointly estimate the launch debris parameter $m$ from the ridge regression. We describe our procedure for calibrating the physical model parameters in more detail in Appendix section \ref{appendix:physical_model_calibration}. Table \ref{tab:calibration_summary} summarizes the calibrated parameter values.

\begin{table}[!htbp]
\caption{Summary of calibrated parameter values for the 600-650 km shell. Values are rounded to the nearest integer or second non-zero decimal place.}
\label{tab:calibration_summary}
\begin{tabular}{lll}
\\[-1.8ex]\hline 
\hline \\[-1.8ex] 
Parameter                  & Value           & Notes    \\  
\hline \\[-1.8ex] 
$\eta_1^{F}$    & 0.025           & Total costs growth parameter. Standard error is 0.009.                       \\
$\gamma_0$                 & 3.35e-06        & Equilibrium adjustment coefficient 1 (open-access capital gains).      \\
$\gamma_1$                 & 2.22e-05        & Equilibrium adjustment coefficient 2 (gross satellite rate of return). \\
$\gamma_2$                 & -2.67e-06       & Equilibrium adjustment coefficient 3 (open-access capital gains).      \\
$\delta$                   & 0.074           & Annual fraction of debris decaying to lower shell.                                                        \\
$\mu$                      & 0.15            & Annual active satellite turnover rate  \\
$\alpha_{SS}$              & 2.73e-07        & Satellite-satellite collision rate parameter.                                    \\
$\alpha_{SD}$              & 2.73e-07        & Satellite-debris collision rate parameter.                                       \\
$\alpha_{DD}$              & 2.78e-07        & Debris-debris collision rate parameter.                                          \\
$\kappa_{SS}$              & 0.99            & Fraction of satellite-satellite collisions successfully avoided.                 \\
$\kappa_{SD}$              & 0.95            & Fraction of satellite-debris collisions successfully avoided.                    \\
$\tilde{\beta}_{SS}$       & 1,800          & Expected number of fragments from satellite-satellite collision. (regularized).   \\
$\tilde{\beta}_{SD}$       & 333           & Expected number of fragments from satellite-debris collision. (regularized).      \\
$\tilde{\beta}_{DD}$       & 327             & Expected number of fragments from debris-debris collision (regularized).         \\
$m$                        & 0.013           & Expected number of launch debris remaining in shell after 1 year (regularized). \\
\hline \\[-1.8ex] 
\end{tabular}
\end{table}

To calculate the intrinsic collision probabilities $\alpha_{SS}, \alpha_{SD}, \alpha_{DD}$, we start with data regarding object cross-sectional areas for active satellites (commercial, military, civil government, and other) and intact debris objects. We assume debris fragments are uniform aluminium spheres of diameter 10 cm, and treat all other objects as uniform spheres as well. We compute the cross-sectional areas of active satellites and debris within each shell as share-weighted averages over 2006 2019 across the types of objects within each class, e.g. if 20\% of the debris objects are intact and 80\% are fragments we calculate the area as $0.2 * (intact~area) + 0.8 * (fragment~area)$. Under these assumptions the rate at which a reference object moving randomly at speed $s$ in a closed space of volume $V$ is struck by an object of cross-sectional area $a$ is
\begin{equation}
    \frac{sa}{V},
\end{equation}

where the volume is determined by the altitude and our assumption of that the space is a spherical shell, and the speed is determined by the altitude, the Earth's gravitational constant, and our assumption that the objects are uniform spheres.

To calculate the unadjusted fragmentation rates, we use data on average object masses from ESA along with a formula found to fit the high-fidelity NASA standard breakup model described in \cite{krisko2011proper}. Letting the mass of the object struck be $M$, and assuming the object is shattered into uniform 10 cm spheres, the number of fragments from a catastrophic collision $n$ is

\begin{equation}
    n = 0.1 M^{0.75} 0.1^{-1.71}.
\end{equation}

The only steps remaining are to adjust our estimate of the expected number of fragments from collisions for the non-random motion of objects in the shell, and to set the value of the launch debris parameter $m$. ODE-based engineering models of the debris environment use such adjustment coefficients based fitting the ODE model to results from many computationally-costly runs of high-fidelity orbital environment models, e.g. as in \cite{somma2017statistical, somma2019adaptive}. This approach would be even costlier for our model, as the launch rate is endogenous, and would not provide a useful estimate of the launch debris parameter $m$. We instead perform the adjustment and estimate $m$ jointly using historical data and ridge regression, a regularization technique used to improve out-of-sample predictive performance at the expense of in-sample fit. Ridge regression achieves this goal by exploiting the bias-variance trade-off, shrinking parameter values toward zero in exchange for reduced prediction variance \citep{hoerl1985practical, zou2005regularization}.

Since satellites are specifically coordinated to reduce collisions, the adjustment for non-random motion should involve shrinking the expected number of fragments from a collision (with the expectation taken over the probability of a collision) toward zero. Ridge regression achieves this goal. Additionally, ridge regression is often used when the number of variables is ``large'' relative to the number of observations or when parameter estimates are known to be noisy due to (for example) high degrees of collinearity. Our model and data satisfy the former condition (with 4 parameters to estimate from 14 observations), and our physical calibration approach (specifically the assumption that all objects are uniform spheres) causes collinearity in our collision probability values.  Since our collision model prescribes the functional form of the collision probability as $(1 - \exp(-{\alpha}_{jk} k))$, the effect of non-random motion on new debris growth cannot be separately identified from $\alpha_{jk}$ and $\beta_{jk}$. This is convenient for our regression-based adjustment, since it allows us to pose the ridge regression as a linear model. Specifically, letting $\bar{x}$ denote a physically-calibrated parameter value, we estimate the following regression:
\begin{align}
\hspace{-1.75cm}
    D_{t+1} - (1 - {\delta})D_t = \rho_{SS} {\beta}_{SS} (1 - \exp(-{\alpha}_{SS} S_t)) + \rho_{SD} {\beta}_{SD} (1 - \exp(-{\alpha}_{SD} S_t)) + \rho_{DD} {\beta}_{DD} (1 - \exp(-{\alpha}_{DD} D_t)) + m + \nu^D_t,
\end{align}

where $\rho_{SS}, \rho_{SD}, \rho_{DD}, m$ are parameters to be estimated and $\nu^D_t$ is the error term. The final regularized estimates of the fragmentation and launch debris parameters are shown in table \ref{tab:calibration_summary} as $\tilde{\beta}_{SS}$, $\tilde{\beta}_{SD}$, $\tilde{\beta}_{DD}$, $m$.

\section{Algorithms for equilibrium and optimum}

To describe how we generate initial guesses for the social planner's problem, it is useful to formally state a finite-horizon sequence version of the planner's problem. Letting $T$ be the final period, the planner's finite-horizon sequence problem is
\begin{align}
    \max_{\{X_t, S_{t+1}, D_{t+1}\}_{t=0}^{T}}& S_t Q(S_t,D_t,X_t) + \frac{1}{1+r} \sum_{\tau=t}^{T-1} \frac{1}{1+r}^{\tau - t - 1}X_{\tau} \left(\frac{1}{1+r} Q(S_{\tau+1},D_{\tau+1},X_{\tau+1}) - F \right) \label{prog:planner-sequence} \\
\text{s.t. }    
    Q(S_t,D_t,X_t) &= \pi + \frac{1}{1+r}(1 - L(S_t,D_t)) Q(S_{t+1}, D_{t+1}, X_{t+1}) \text{ if } t<T \nonumber\\
    Q(S_T,D_T,X_T) &= \pi \nonumber\\
    S_{t+1} &\leq S_t(1-L(S_t,D_t)) + X_t \nonumber\\
    D_{t+1} &\geq D_t(1-\delta) + G(S_t,D_t) + m X_t\nonumber\\
    X_t &\in [0,\bar{X}] ~~ \forall t \nonumber\\
    S_{t+1} & \geq 0, D_{t+1} \geq 0 \nonumber\\
    S_0 &= s_0, D_0 = d_0. \nonumber
\end{align}

The guess generation in algorithm \ref{algo:planners-problem} uses a result from \citet{easley1981stochastic}, that optimal plans generated from solving a finite horizon problem with sufficiently-large $T$ closely approximate infinite-horizon optimal plans. Algorithm \ref{algo:planners-problem} describes our solution procedure more precisely.

\begin{algorithm}
	Generate a sparse initial grid, $\mathcal{G}_0$, over $(S,D) \in \mathbb{R}_{[0,a]}\times\mathbb{R}_{[0,b]}, ~~ a,b>0$.
	
	At each point on $\mathcal{G}_0$, solve program \ref{prog:planner-sequence} with $T$ equal to a large number. Larger is better; we use $T=150$, which balances compute time with guess quality. This produces an initial guess on a sparse grid, $\tilde{v}_0$.
	
	Using linear interpolation, ``infill'' $\tilde{v}_0$ (defined on $\mathcal{G}_0$) to $v_0$ (defined on $\mathcal{G}_1$). $\mathcal{G}_1$ has the same boundaries as $\mathcal{G}_0$ ($(S,D) \in \mathbb{R}_{[0,a]}\times\mathbb{R}_{[0,b]}$) but contains more points. This gives an initial guess defined on a denser grid.
	
	Set $\delta$ to some large number (we use 10) and $\epsilon$ to some small number (we use 1\% of the mean value of $v_0$). Set $i=0$ and $W_0(S,D) = v_0$.
	
	\While{$ \delta > \epsilon$}{
		At each node in $\mathcal{G}_1$, solve program \ref{prog:planner-dp} with $W(S_{t+1},D_{t+1}) = W_i(S_{t+1},D_{t+1})$. Label the value function obtained as $W_{i+1}(S,D)$, defined over $\mathcal{G}_1$. We use linear interpolation to compute $W_i(S_{t+1},D_{t+1})$ when $(S_{t+1},D_{t+1})$ is between nodes of $\mathcal{G}_1$.
		
		$\delta \leftarrow || W_i(S,D) - W_{i+1}(S,D)||_{\infty}$.
		
		i $\leftarrow$ i+1	
	}
	\caption{Solve the planner's problem}\label{algo:planners-problem}
\end{algorithm}

Generating the open-access policy function is much simpler. At each node on a grid over $S$ and $D$ values (e.g. $\mathcal{G}_1$ as in Algorithm \ref{algo:planners-problem}), we solve the open-access condition
\begin{align}
    \pi - rF - L(S'(X,S,D),D'(X,S,D))F = 0
\end{align}
for the open-access launch rate $X$.

To generate the phase diagrams, we use solved policy functions $X$ to compute the evolution of the satellite and debris stocks at each grid node. More precisely, we compute $dS = (S'(X,S,D) - S)/h$ and $dD = (D'(X,S,D) - S)/h$ for a fixed positive value $h$. The value of $h$ is chosen to make the plotting more stable; we use $h=10$, but other values yield similar results. The nullclines are plotted as the zero-isoclines of $dS$ and $dD$.

\end{appendices}

\end{document}